\documentclass[journal]{IEEEtran}
\usepackage{cite}
\usepackage{mathrsfs}
\usepackage{amssymb}
\usepackage{amsmath}
\usepackage{amsthm}
\usepackage{graphicx}
\usepackage{tikz}
\usetikzlibrary{arrows.meta, positioning}
\usepackage{booktabs}
\usepackage{subcaption}  % 提供 subfigure 环境
\usepackage{threeparttable}

\newtheorem{theorem}{Theorem} % <-- 去掉了 [section]，实现全局编号
\newtheorem{definition}{Definition} % <-- 同样共享计数器
\ifCLASSINFOpdf
\else
\fi
\begin{document}
%
% paper title
% Titles are generally capitalized except for words such as a, an, and, as,
% at, but, by, for, in, nor, of, on, or, the, to and up, which are usually
% not capitalized unless they are the first or last word of the title.
% Linebreaks \\ can be used within to get better formatting as desired.
% Do not put math or special symbols in the title.
\title{ A Quantum Walk-Enabled Blockchain with Weighted Quantum Voting Consensus}
%A Feasible Quantum Walk-Enabled Blockchain with Quantum Delegated Proof-of-Stake Consensus
% 
% author names and IEEE memberships
% note positions of commas and nonbreaking spaces ( ~ ) LaTeX will not break
% a structure at a ~ so this keeps an author's name from being broken across
% two lines.
% use \thanks{} to gain access to the first footnote area
% a separate \thanks must be used for each paragraph as LaTeX2e's \thanks
% was not built to handle multiple paragraphs
%

\author{Chong-Qiang~Ye, Heng-Ji~Li, Jian~Li and Xiao-Yu~Chen
        
\thanks{Chong-Qiang Ye  and Xiao-Yu~Chen are with School of Information and Electrical Engineering, Hangzhou City University, Hangzhou 310015, China, e-mails: (chongqiangye@hzcu.edu.cn, chenxiaoyu@hzcu.edu.cn).}% <-this % stops a space
\thanks{Heng-Ji~Li is with School of Artificial Intelligence, Henan University Zhengzhou, China, e-mail: (lihengji@henu.edu.cn).}
\thanks{Jian Li is with School of Cyberspace Security, Beijing University of Posts and Telecommunications, Beijing 100876, China, e-mail: (jianli@bupt.edu.cn).}
%\thanks{This work was supported in part by  the National Natural Science Foundation of China under Grant 62501523.}
}

% <-this % stops a space
%\thanks{Manuscript received April 19, 2005; revised August 26, 2025.}}

% note the % following the last \IEEEmembership and also \thanks - 
% these prevent an unwanted space from occurring between the last author name
% and the end of the author line. i.e., if you had this:
% 
% \author{....lastname \thanks{...} \thanks{...} }
%                     ^------------^------------^----Do not want these spaces!
%
% a space would be appended to the last name and could cause every name on that
% line to be shifted left slightly. This is one of those "LaTeX things". For
% instance, "\textbf{A} \textbf{B}" will typeset as "A B" not "AB". To get
% "AB" then you have to do: "\textbf{A}\textbf{B}"
% \thanks is no different in this regard, so shield the last } of each \thanks
% that ends a line with a % and do not let a space in before the next \thanks.
% Spaces after \IEEEmembership other than the last one are OK (and needed) as
% you are supposed to have spaces between the names. For what it is worth,
% this is a minor point as most people would not even notice if the said evil
% space somehow managed to creep in.

% The paper headers
%\markboth{Journal of \LaTeX\ Class Files,~Vol.~14, No.~8, October~2025}%
%{ Ye \MakeLowercase{\textit{et al.}}: A Feasible Quantum Walk-Enabled Blockchain with Quantum Delegated Proof-of-Stake Consensus}

\maketitle

% As a general rule, do not put math, special symbols or citations
% in the abstract or keywords.
\begin{abstract}
Quantum blockchains provide inherent resilience against quantum adversaries and represent a promising alternative to classical blockchain systems in the quantum era. However, existing quantum blockchain architectures largely depend on entanglement to maintain inter-block connections, facing challenges in stability, consensus efficiency, and system verification. To address these issues, this work proposes a novel quantum blockchain framework based on quantum walks, which reduces reliance on entanglement while improving stability and connection efficiency.  We further propose a quantum consensus mechanism based on a weighted quantum voting protocol, which enables a fairer voting process while reflecting the weights of different nodes. To validate the proposed framework, we conduct circuit simulations to evaluate the correctness and effectiveness of both the quantum walk-based block construction and the quantum voting consensus mechanism. Compared with existing entanglement-dependent approaches, our framework achieves stronger stability and enables simpler verification of block integrity, making it a practical candidate for quantum-era blockchain applications.
\end{abstract}

% Note that keywords are not normally used for peerreview papers.
\begin{IEEEkeywords}
Quantum blockchain,  quantum walk,  weighted quantum voting, circuit simulations.
\end{IEEEkeywords}

% For peer review papers, you can put extra information on the cover
% page as needed:
% \ifCLASSOPTIONpeerreview
% \begin{center} \bfseries EDICS Category: 3-BBND \end{center}
% \fi
%
% For peerreview papers, this IEEEtran command inserts a page break and
% creates the second title. It will be ignored for other modes.
\IEEEpeerreviewmaketitle

\section{Introduction}

\IEEEPARstart{W}{ith}  the rapid development of quantum computing technology, it has posed a disruptive challenge to traditional computing paradigms and information security systems. Leading technology companies such as Google and IBM have announced significant milestones in achieving quantum supremacy, marking a critical transition from theoretical research to practical implementation~\cite{1,2}. One of the core threats posed by quantum computing lies in Shor's algorithm, which can efficiently solve integer factorization and discrete logarithm problems in polynomial time~\cite{3}. This directly undermines the security assumptions underlying classical public-key cryptographic systems such as RSA and ECC~\cite{4}.

Blockchain technology, as a representative of distributed trust systems, has seen widespread adoption in key sectors including finance, healthcare, and governmental record-keeping, all of which rely heavily on the security of classical cryptography~\cite{5,6}. However, once scalable quantum computers become practical, existing blockchain systems will face catastrophic risks such as transaction tampering and asset theft, threatening the security and stability of the global digital economy. Therefore, exploring blockchain architectures with resistance to quantum attacks is a necessary step toward maintaining the security of future information systems.

To address the potential threats that quantum computing poses to traditional cryptographic primitives and blockchain security, quantum-resistant blockchain technologies have become a prominent research topic~\cite{7}. Current developments in this domain can be broadly categorized into three main approaches:

{\it Post-quantum cryptography-based blockchains}: These systems replace traditional public-key primitives with post-quantum algorithms such as lattice-based cryptography~\cite{8}, hash-based signatures~\cite{9}, and multivariate cryptosystems~\cite{10}. Operating entirely within the classical computational framework, post-quantum cryptography (PQC)-based blockchains significantly enhance resistance against quantum attacks while maintaining compatibility with existing infrastructures~\cite{11}. Owing to this deployability and interoperability, PQC is currently regarded as the most practical route toward quantum-resistant blockchain systems. Nevertheless, the reliance on classical architectures and the relatively large key and signature sizes introduce computational and storage overhead. Thus, while PQC-based solutions play a crucial role in providing security during the transition to the quantum era, they inherently remain bounded by classical cryptographic assumptions and cannot offer the same fundamental security guarantees as blockchain systems leveraging quantum technologies.

{\it Hybrid quantum-cryptographic blockchains}: These architectures selectively incorporate quantum technologies, such as quantum key distribution (QKD)~\cite{12}, quantum signatures~\cite{13}, and quantum teleportation~\cite{14} into classical blockchain frameworks. By enhancing specific components with quantum-level security, hybrid systems can mitigate particular classes of quantum threats. For example, Kiktenko et al.~\cite{15} proposed a quantum-secured blockchain in which classical digital signatures are replaced with QKD-generated keys, thereby enabling decentralized block generation under a quantum-enhanced security model. While such approaches strengthen certain security aspects, they remain fundamentally grounded in classical infrastructures, which constrains scalability and limits their seamless integration into prospective quantum-native blockchain environments.

{\it Quantum blockchains}:  In contrast, quantum blockchains are conceived as fully quantum-native systems, reconstructing fundamental blockchain elements such as storage, verification, and consensus on the basis of quantum principles~\cite{16,17,18,19}. For example, the model proposed by Rajan and Visser~\cite{16} encodes blockchain data via the time evolution of entangled quantum states, thereby offering inherent resistance to tampering. A notable advantage of this approach is its reliance on entanglement, whereby any attempt to modify the stored data collapses the underlying quantum state and reveals the manipulation. By directly exploiting the principles of quantum mechanics, such designs offer strong theoretical security guarantees, even in the presence of quantum-capable adversaries. Continued theoretical investigation of quantum blockchain technologies is therefore of critical importance, as it lays the groundwork for future secure and scalable distributed systems.

Among various approaches to quantum-resistant blockchain design, fully quantum-based systems are considered the most innovative and forward-looking, as they aim to overcome the inherent limitations of classical architectures and offer a foundation for next-generation secure distributed ledgers\cite{20}. Despite its great potential, quantum blockchain faces several significant challenges, including the need for scalable quantum hardware and robust quantum entanglement network infrastructure~\cite{21}. In particular, existing quantum blockchain architectures typically rely on entanglement to maintain correlations between blocks. However, sustaining long-lived, high-fidelity entanglement demands substantial quantum resources, such as quantum repeaters and entanglement swapping, which remain technically difficult to implement in large-scale distributed environments~\cite{22,23}. Moreover, current quantum consensus protocols heavily depend on highly entangled quantum states, which significantly increases quantum resource consumption and complicates practical deployment. Meanwhile, delegated consensus mechanisms based on quantum voting protocols~\cite{24,25,26}, such as quantum delegated proof of stake, remain limited in supporting weighted participation and achieving scalability. Developing quantum consensus frameworks that can accommodate weighted voting and dynamic participation remains a critical challenge.

To address these limitations, there is a pressing need for alternative designs that reduce reliance on entanglement and support the development of practical quantum-secure consensus mechanisms. In this context, quantum walks~\cite{27,28,29} present a promising solution. As a fundamental model in quantum computing, quantum walks are characterized by structured and controllable evolution, enabling the construction of lightweight quantum protocols that minimize the need for demanding global resources, such as extensive entanglement. Their dependence on local operations and single-qubit control makes them particularly compatible with near-term quantum technologies~\cite{30}, and they may provide a practical pathway toward implementing quantum blockchain systems.

Motivated by these advantages, this work proposes a novel quantum blockchain framework that leverages quantum walks to establish inter-block connectivity, thereby reducing entanglement requirements while enhancing system scalability and stability. Then, we present a quantum private permutation protocol based on high-dimensional entangled states and integrated it with a quantum secure multi-party summation protocol to construct a quantum delegated proof-of-stake consensus mechanism that supports differentiated voting weights. To evaluate the proposed framework, we conducted simulation experiments using quantum circuits on a quantum cloud platform, demonstrating its functional correctness and potential for efficient deployment in real-world quantum environments.

The remainder of this paper is organized as follows. Related work is reviewed in Section 2. Section 3 presents the proposed quantum blockchain framework based on quantum walks, along with the quantum consensus protocol. The security analysis and simulation results are presented in Section 4, while Section 5 provides comparisons and discussion. Finally, Section 6 concludes the paper.

\section{Related work }
In this section, we mainly introduce preliminaries and related work about  quantum blockchains and quantum walks.

\subsection{Quantum blockchains}

\subsubsection{Quantum Foundations for Blockchain Systems}
Quantum computing introduces transformative principles for information processing, leveraging quantum bits (qubits) that differ fundamentally from classical bits. Unlike classical bits, which represent either 0 or 1, qubits can exist in a superposition of states, described as $\alpha|0\rangle + \beta|1\rangle$, where $\alpha$ and $\beta$ are complex amplitudes. Additionally, qubits can be entangled, creating strong correlations that persist across distances. These properties enable unparalleled parallelism and security mechanisms unattainable in classical systems~\cite{31}.

In blockchain contexts, quantum properties offer robust solutions for data integrity and security. The no-cloning theorem~\cite{32}, which prohibits the exact replication of an unknown quantum state, can prevent unauthorized data copying, enhancing trust in decentralized ledgers. Quantum entanglement facilitates tamper-evident structures by linking data through correlated states, detectable upon unauthorized interference. Early designs of quantum blockchains often adopted phase encoding techniques, such as rotation gates $R(\theta)$, to embed classical data into quantum states~\cite{19,26}. A representative encoding scheme maps a classical value $p$ to a superposition state of the form $\frac{1}{\sqrt{2}}(|0\rangle + e^{i\theta_p}|1\rangle)$, where $\theta_p$ encodes the payload. To ensure structural coherence across blocks, the phase angles $\theta_{p_i}$ of the $i$-th block are defined relative to the first block’s phase $\theta_{p_1}$ as
\begin{equation}
\theta_{p_i} = \frac{1}{q^{i-1}} \theta_{p_1},
\end{equation}
where $q > 1$ is a scaling factor that maintains a hierarchical relationship among blocks. This encoding preserves data integrity while leveraging quantum properties for verification. Fig. 1 illustrates the basic structure of a quantum blockchain.

\begin{figure}[h]
    \centering
    \includegraphics[width=3.4in]{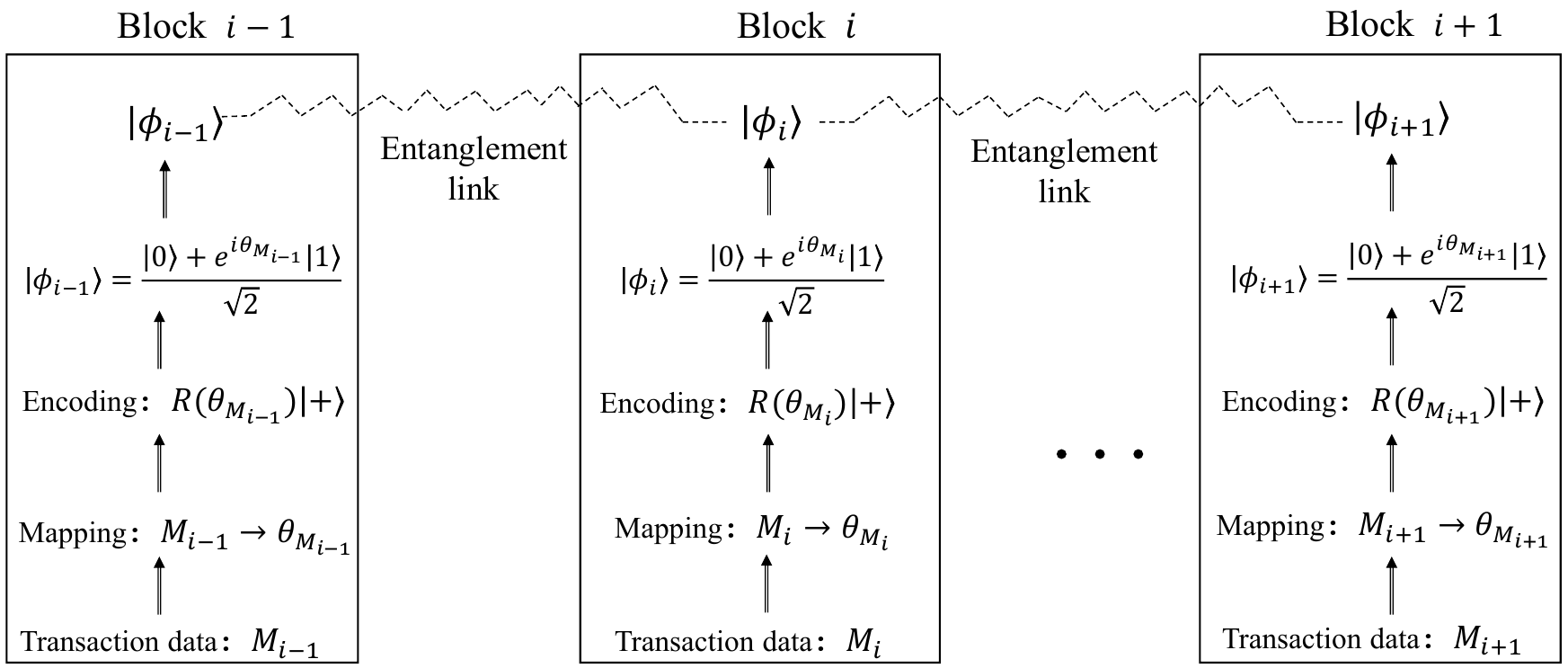}
    \caption{Schematic of a quantum blockchain, showing quantum blocks encoding classical data via phase angles and linked through entangled states.}
    \label{fig:quantum_blockchain}
\end{figure}

In terms of block connectivity, quantum blockchain models typically employ quantum entanglement to replace traditional hash-based linking methods, leveraging quantum correlations to establish block data associations, thereby significantly enhancing system security. For example, Banerjee et al.~\cite{19} introduce a model based on weighted hypergraph states, which provide a flexible structure for distributed ledgers. In this model, blockchain data are encoded into the phases of quantum states via specific unitary operations, and blocks are linked through quantum entanglement. Quantum nonlocality is then employed to ensure tamper resistance. In addition to hypergraph-based designs, other approaches employ multi-particle entangled states, such as GHZ states~\cite{16, 17,18} to preserve global coherence across the ledger, thereby ensuring consistency in decentralized networks.

%Also, other approaches utilize Bell states~\cite{16} or Greenberger-Horne-Zeilinger (GHZ) states~\cite{17,18} to maintain global coherence across the ledger, ensuring consistency in decentralized networks.

Despite their theoretical promise, entanglement-based architectures face significant challenges. Maintaining long-range entanglement in distributed systems requires quantum repeaters and entanglement swapping protocols, which are technologically demanding~\cite{33}. Furthermore, entanglement is highly susceptible to noise and decoherence, limiting the scalability of such systems with current hardware. Consequently, entanglement-dependent blockchains are unlikely to be practical in the near term.

\subsubsection{Quantum Consensus Mechanisms} 
In classical systems, consensus mechanisms such as Proof-of-Work (PoW) and Proof-of-Stake (PoS) rely on computational or economic incentives to achieve consensus. Quantum consensus mechanisms, however, leverage quantum mechanical principles to enhance security, efficiency, and scalability in quantum blockchain systems.

Quantum consensus protocols\cite{34,35,36,37} often exploit the no-cloning theorem and quantum entanglement to secure agreement processes. For instance, Qu et al.~\cite{35} proposed an $n$-party quantum detectable Byzantine agreement protocol based on GHZ states to achieve data consensus in quantum blockchains. Their protocol exploits the non-locality of GHZ states, distributing entangled resources across nodes to verify node integrity and ensure robust consensus. Similar to classical Byzantine agreement protocols, quantum Byzantine approaches encounter scalability challenges as the number of participating nodes increases. Specifically, both the communication complexity and computational overhead grow with network size, which may constrain the protocol's applicability in large-scale systems.

 Another representative work by Li et al.\cite{26} introduced a quantum delegated proof of stake (QDPoS) mechanism that integrates high-dimensional multipartite entangled states with quantum Fourier transforms to enable decentralized voting. Their protocol supports the quantum representation of three voting outcomes—approval, abstention, and disapproval. While Li et al.'s work offers a promising direction for incorporating quantum techniques into delegated consensus, their protocol adopts a uniform voting model, treating all participants as having equal voting power. This abstraction simplifies the protocol but does not fully reflect the weighted voting characteristics of practical DPoS systems, where stakeholders typically possess varying degrees of influence. The lack of a native quantum representation for voting weights may thus limit the protocol’s ability to model realistic delegate selection processes. To bridge this gap, it is important to explore how weighted voting can be integrated into QDPoS within a quantum-secure framework. 

%while still preserving vote privacy and integrity. 

%Addressing this challenge forms the core motivation of our work.

\subsection{Quantum walks}

Here, we briefly introduce the fundamental model and concepts of quantum walks.
Quantum walk-based systems can be categorized into two distinct types: discrete-time quantum walks (DTQWs), characterized by evolution in discrete steps, and continuous-time quantum walks (CTQWs), governed by continuous evolution under a fixed Hamiltonian. In this work, we focus on constructing a quantum blockchain based on the discrete-time quantum walk model.

\subsubsection{Two-direction discrete quantum walks on a circle}
A discrete quantum walk comprises two subsystems~\cite{38}: the walker system, which describes the particle's state in position space, and the coin system, which determines the particle's evolution direction at each step. The corresponding Hilbert spaces for the walker and the coin are denoted as $\mathcal{H}_w$ and $\mathcal{H}_c$, respectively. The total state space of the quantum walk system is given by the tensor product of these two spaces:
$\mathcal{H} = \mathcal{H}_w \otimes \mathcal{H}_c$.
A two-direction discrete quantum walk on a circle system can be expressed as:
\begin{equation}
|\psi\rangle = \sum_{x \in \mathbb{Z}\cap [0, d)} \sum_{c \in \{0,1\}} \psi_{x,c} |x\rangle_w \otimes |c\rangle_c,
\end{equation}
where $\psi_{x,c} \in \mathbb{C}$ represents the complex amplitude of the state $|x\rangle_w \otimes |c\rangle_c$. 

The evolution operator for the quantum walk state is defined as:
\begin{equation}
U = S (I_w \otimes C),
\end{equation}
where $C$ is the coin operator acting on the coin space, and $S$ is the controlled shift operator acting on the walker-coin system.  According to Ref.\cite{38}, $C$ is defined as 
\begin{equation}
C=\begin{pmatrix}
e^{i\xi}\ \text{cos}\ \theta & e^{i\eta}\ \text{sin}\ \theta\\
e^{-i\eta}\ \text{sin}\ \theta & -e^{-i\xi}\ \text{cos}\ \theta\\
\end{pmatrix}.
\end{equation}
For simplicity, we set $\xi = 0$, $\theta = \frac{4}{\pi}$, and $\eta = 0$ in this work. Under these parameter choices, the coin operator $C$ reduces to a Hadamard gate, i.e.,  $ C=H=\frac{1}{\sqrt{2}} \left( \begin{array}{cc}
1 & 1 \\
1 & -1 \\
\end{array}\right)$. $S$ is defined as:
\begin{equation}
S =  T_1\otimes |0\rangle_c \langle 0|_c + T_{-1}\otimes |1\rangle_c \langle 1|_c ,
\end{equation}
where $T_1$ and $T_{-1}$ are the shift operators over a $d$-dimensional position space:
\begin{equation}
T_1=\sum^{d-1}_{x=0}|x\oplus 1\rangle \langle x|, \quad  T_{-1}=\sum^{d-1}_{x=0}|x\oplus -1\rangle \langle x|.
\end{equation}
Here the symbol $\oplus$ represents the modulo $d$ addition. % Fig.2 illustrates the quantum circuit representations of $T_1$ and $T_{-1}$.
\begin{figure}[h]
    \centering
    \includegraphics[width=3.3in]{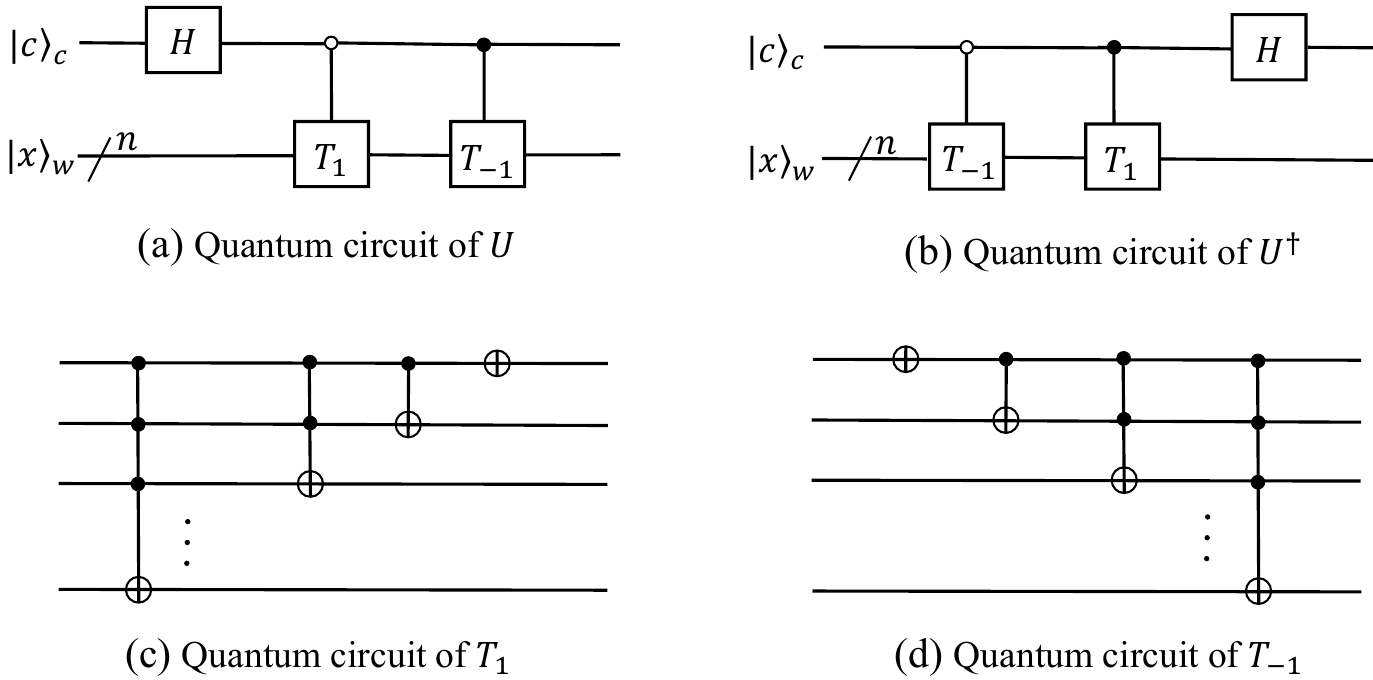}
    \caption{Circuit diagram of quantum walk evolution operation.}
\end{figure}
After defining the evolution operator $U$ and its components, the inverse evolution operator is correspondingly given by $U^\dagger = (I_w \otimes C^\dagger) S^\dagger$.  The quantum circuits for $U$ and $U^\dagger$ are shown in Fig. 2(a) and (b), with the detailed constructions of the shift operators $T_1$ and $T_{-1}$ provided in Fig. 2(c) and 2(d).
Accordingly, the evolution of the quantum walk can be described as follows. Given an initial state $|\psi_0\rangle$, the quantum state after $t$ steps of the quantum walk evolves to:
\begin{equation}
|\psi_t\rangle = U^t |\psi_0\rangle.
\end{equation}
The intrinsic sequential dependency of a quantum walk is similar to blockchain linkage. By encoding each block’s data into its corresponding step operator, the state obtained after $t$ steps depends on the initial state (i.e.,  the data of prior blocks), producing a blockchain‑like chained structure (see Fig. 3).
\begin{figure}[h]
    \centering
    \includegraphics[width=3.4in]{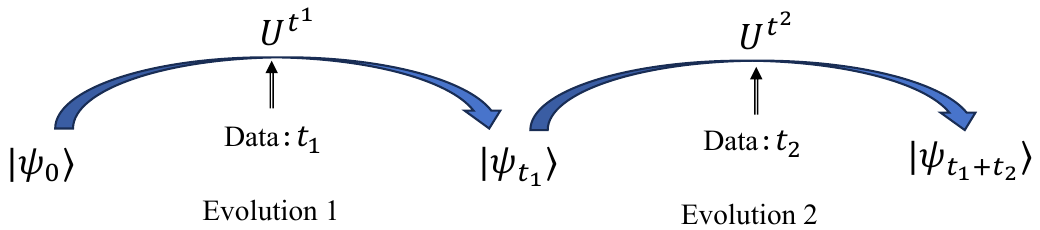}
    \caption{Schematic diagram of chain-structured state evolution based on quantum walk.}
\end{figure}

\subsubsection{Quantum hash function based on quantum walks}
Quantum hash functions (QHFs) \cite{39,40,41,42} have garnered significant attention for their potential to enhance security in quantum cryptographic protocols. Among various constructions, QHFs based on quantum walks are particularly notable for their intrinsic randomness and pronounced sensitivity to initial conditions and control parameters \cite{41,42}.  By exploiting the inherent sensitivity and dynamic properties of quantum walks, these constructions can achieve strong cryptographic properties such as collision resistance and unpredictability.

For example, Ref.\cite{42} presents a QHF construction based on a discrete-time quantum walk involving two particles on a cycle. In this scheme, the input message, represented as a classical binary string, determines the sequence of coin operators applied at each step of the quantum walk evolution. Specifically, two different coin operators are defined: the Grover coin $C_0$ and an alternative coin $C_1$. 
\begin{equation}
\resizebox{.9\hsize}{!}{$
C_0 = \frac{1}{2}
\begin{pmatrix}
-1 & 1 & 1 & 1 \\
1 & -1 & 1 & 1 \\
1 & 1 & -1 & 1 \\
1 & 1 & 1 & -1 \\
\end{pmatrix},\ 
C_1 = \frac{1}{2}
\begin{pmatrix}
1 & 1 & 1 & 1 \\
1 & 1 & -1 & -1 \\
1 & -1 & 1 & -1 \\
1 & -1 & -1 & 1 \\
\end{pmatrix}.
$}
\end{equation}
For each bit of the input message, the corresponding coin operator is chosen, typically $C_0$ for 0 and $C_1$ for 1. The construction process of the QHF is as follows\cite{42}:

 (1) Select the parameters $n$ (the number of nodes on the cycle) and $(\alpha, \beta, \chi, \delta)$, where $\alpha, \beta, \chi, \delta$ are the amplitudes of the initial coin state $|c\rangle = \alpha|00\rangle + \beta|01\rangle + \chi|10\rangle + \delta|11\rangle$. %The input message can be of arbitrary length.
 
(2) Under the control of the message, run the two-particle discrete-time quantum walk on a cycle. The evolution of the quantum walker is governed by the sequence of coin operators determined by the input message.

(3) After a fixed number of steps, extract the probability distribution of the walker's position. All values in the resulting probability distribution are multiplied by $10^8$ and reduced modulo 256 to form a binary string as the secret key $hv$, i.e., the hash value.

This approach ensures that even a slight change in the input message leads to a significant alteration in the quantum walk evolution and the resulting output.

% thereby providing the desired sensitivity and collision resistance for the hash function. 
 
\section{Quantum Blockchain Construction via Quantum Walks}

In this section, we propose a novel quantum blockchain architecture based on discrete-time quantum walks. Each quantum block consists of multiple quantum walk states, with initial states determined by the hash of the previous block. Transactions and metadata are encoded into the quantum walk steps, which evolve to form the block state, inherently linking it to both the previous block and the current data, thus ensuring blockchain integrity. A weighted quantum voting-based delegated proof of stake mechanism is then used to select representatives, achieve consensus, and validate candidate blocks. Upon consensus, the verified transactions are incorporated into quantum blocks and synchronized across all nodes. The workflow for each phase is detailed below.

% with the overall process illustrated in Fig. 4.

\subsection{Quantum block structure construction}
 In this phase, classical block information is encoded into quantum states via discrete-time quantum walks, as illustrated in Fig. 4. Each quantum block body contains $n$ independent discrete-time quantum walk states, denoted as $|\psi_1\rangle,|\psi_2\rangle,\dots,|\psi_n\rangle$. The block data is first mapped to a set of walk step counts, which are then used to drive the quantum walk evolution for each state. The full procedure for initialization, encoding, and evolution is as follows.

\subsubsection{Initialization of position states}  For the $i$-th block $B_i$, the initial position states of the quantum walkers are derived from the hash value $hv_{i-1}$ of the previous block $B_{i-1}$'s metadata, which includes transaction data, timestamp, and other relevant fields. Specifically, the hash $hv_{i-1}$ is divided into $n$ segments, each of length $L=\text{log}_2 M $, where $M$ denotes the size of the position space. Each segment is then interpreted as an integer $x^i_j\in [0,M-1]$,  serving as the initial position for the $j$-th walker. Consequently, the initial quantum walk states are given by:
\begin{equation}
|\psi^i_j\rangle =  \sum_{x} \sum_{c} |x^i_j\rangle\otimes|c^i_j\rangle, \quad \text{for}\  j=1,2,\dots,n.
\end{equation}

\subsubsection{Mapping of block information}  
To encode the current block's data into quantum walk dynamics, the transaction data and timestamp of block $B_i$ are first concatenated and divided into $n$  equal-length segments. Each segment is interpreted as a binary string and then converted into a non-negative integer $s^i_j$, for $j = 1, 2, \dots, n$. These integers are used to determine the number of walk steps for each quantum walker.

To introduce randomness and prevent deterministic encoding, each step count $t^i_j$ is computed as:
\begin{equation}
t^i_j = (s^i_j + r^i_j) \bmod T + 1,
\end{equation}
where $r^i_j $ is a randomization factor derived from the lower-order bits of the block’s timestamp, and $T$ is the system-defined upper bound for the number of walk steps.

\begin{figure}[h]
    \centering
    \includegraphics[width=3.5in]{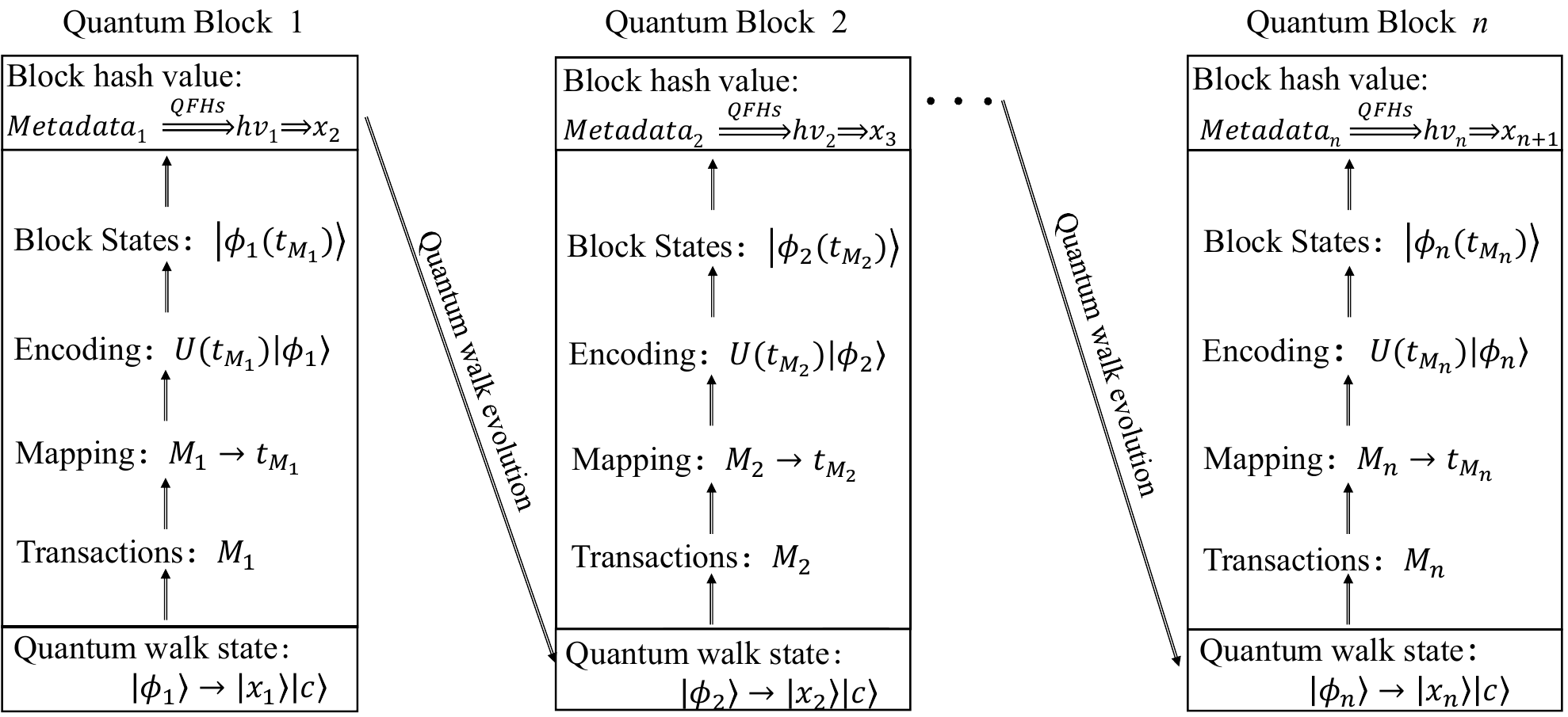}
    \caption{Schematic of the quantum block structure based on quantum walks. Each block corresponds to a quantum state $|\phi\rangle$, composed of $n$ quantum walk states $\lvert\psi_1\rangle, \lvert\psi_2\rangle, \dots, \lvert\psi_n\rangle$. The initial positions of these walks are derived from the hash of the preceding block, while the number of walk steps encodes the information of the current block. }
\end{figure}
 
 %Through this process, the quantum walk evolution naturally links the data and establishes the structural connection between consecutive blocks.

\subsubsection{Data embedding via quantum walk evolution}  
 
Once the step counts $t^i_j$ are determined, each quantum walker evolves under the discrete-time quantum walk operator accordingly. Specifically, for the $j$-th walker in block $B_i$, the initial state $|\psi^i_j\rangle $ evolves through $t^i_j $ steps via repeated application of the unitary walk operator $ U $, resulting in:
\begin{equation}
|\psi^i_j(t^i_j)\rangle = U^{t^i_j} |\psi^i_j\rangle.
\end{equation}
This evolution process embeds the current block's transaction and timing information into the resulting quantum walk states.

The complete set of evolved states \( \{|\psi^i_j(t^i_j)\rangle\}_{j=1}^n \) collectively represents the quantum body of block \( B_i \). These states are intricately linked with both the previous block (via the initial positions derived from the previous hash) and the current block data (via the encoded steps), ensuring forward and backward integrity in the blockchain structure.

\subsubsection{Computation of block hash} Finally, a fixed-length hash value $ hv_i $ is computed from the current block’s transaction data and timestamp using a quantum hash function \cite{42}. The resulting hash $ hv_i$ is stored in the block header and serves as the unique identifier for block $B_i $. Importantly, $hv_i $ also acts as the basis for initializing the quantum position states of the next block $B_{i+1}$, thereby maintaining the forward-linking structure of the blockchain. This approach ensures quantum-level integrity between consecutive blocks, ensuring that each block is cryptographically linked to its predecessor through quantum walk-based evolution and hash generation.

\subsection{Quantum consensus mechanism  }
In this subsection, we propose a novel quantum delegated proof of stake consensus (QDPoS), which integrates a weighted quantum voting protocol into a delegated proof-of-stake framework. Unlike prior schemes that employ quantum voting to select representatives, our method introduces weighted voting, where each participant's voting power is proportional to their stake or contribution in the system. This enables a more fine-grained and fair representation of user influence in the consensus process. 
The core quantum resource employed in this protocol is the $n$-dimensional $n$-particle Cat state, expressed as\cite{43}:
\begin{equation}
|\Phi(\vartheta_1,\dots,\vartheta_n)\rangle = \frac{1}{\sqrt{n}} \sum_{l=0}^{n-1} \omega^{l \vartheta_1} |l, l+\vartheta_2, \dots, l+\vartheta_n\rangle,
\end{equation}
where $\omega=e^{\frac{2 \pi i}{n}}$ and the label $\vartheta_l \in [0,n-1]$. By setting $\vartheta_1 = 0$ and ensuring the remaining parameters $(\vartheta_2, \dots, \vartheta_n)$ are mutually distinct, the state exhibits two key measurement properties, as demonstrated in Ref.~\cite{43}: (1) In the computational basis, the measurement outcomes for the $n$ qudits are all unique. (2) In the Fourier basis, the sum of all measurement outcomes is zero, modulo $n$. Based on these properties, the proposed QDPoS consensus mechanism proceeds through the following phases.

\subsubsection{Initialization and parameter setup} 
Suppose there are $m$ candidates competing to be elected as representative nodes, and the goal is to select $r$ representatives from them, denoted as $R_1, R_2, \dots, R_r$. The set of eligible voters is given by $ \{V_0, V_1, \dots, V_{n-1}\}$, where each voter $V_l$ is associated with a predefined voting weight $w_l \in \mathbb{R}^{+} $, for $l = 0, 1, \dots, n-1$. These weights reflect the relative stake or contribution of each voter within the network, and are normalized and discretized to determine the number of votes each participant can cast. Formally, let $W = \sum_{l=0}^{n-1} w_l$ denote the total system-wide stake. Then, the effective voting share of voter $V_l$ is given by $
\tilde{w}_l = \left\lfloor \frac{w_l}{W} \cdot T_v \right\rfloor,
$
where $T_v $ is a fixed upper bound on the total number of votes to be distributed, and $ \tilde{w}_l \in \mathbb{Z}^{+} $ represents the quantized voting weight of $ V_l$. %Each voter then encodes their voting preference using a quantum voting scheme described in the subsequent sections.

\subsubsection{Privacy index distribution}
At this stage, quantum entanglement is utilized to distribute voter privacy indices, as detailed in the following steps.

{\it Step 1}: Each candidate $C_k\ (k=1,2,\dots,m)$ first prepares $1+\delta$ group Cat states $|\Phi(0,\vartheta^k_2,\dots,\vartheta^k_n)\rangle$, with each parameter set $(\vartheta^k_2, \dots, \vartheta^k_n)$ mutually distinct both within and across groups. Each group consists of two identical copies, and the $l$-th particle of each state is assigned to the voter $V_l$. 
% Specifically, the prepared particle sequences form the following matrix:
%\begin{equation}
%\begin{pmatrix}
%P_{1,1} &\dots  &  P_{1,i} &\dots & P_{1,n}\\
%\vdots &\vdots  &  \vdots  &\vdots & \vdots\\
%P_{t^\prime,1} &\dots  &  P_{t^\prime,i} &\dots & P_{t^\prime,n}\\
%\vdots &\vdots  &  \vdots  &\vdots & \vdots\\
%P_{2(1+\delta),1} &\dots  &  P_{2(1+\delta),i} &\dots & P_{2(1+\delta),n}
%\end{pmatrix},
%\end{equation}
%where $1\le t^\prime\le 2(1+\delta)$ and each column $i\in\{1,2,\dots,n\}$ corresponds to the sequence of particles assigned to voter $V_i$. 

{\it Step 2}: To ensure secure and anonymous index distribution, each candidate mixes the prepared states with decoy particles randomly selected from both the computational basis and the Fourier basis. The resulting sequences are then transmitted to the corresponding voters. 

{\it Step 3}: After receiving the particle sequences, each voter first performs a round of security verification using the inserted decoy particles to check the security of the transmission channel. If the detected error rate is below a predefined threshold, the remaining particles are considered valid and retained for subsequent steps in generating private voting indices.

{\it Step 4}:  To ensure that the candidate has correctly prepared the Cat states, all voters collaboratively select $\delta$ groups from the $1+\delta$ prepared Cat state pairs for verification. The candidate $C_k$ is then required to publicly disclose the parameters used in preparing each selected Cat state, i.e., $(\vartheta^k_2,\vartheta^k_3,\dots,\vartheta^k_n)$. For each selected group, which contains two identical copies of the same Cat state, the voters perform a computational-basis measurement on one copy and a Fourier-basis measurement on the other. According to the properties of Cat states: In the computational basis, all voters should observe mutually distinct outcomes; In the Fourier basis, the sum of measurement results (modulo $n$) should be zero.
If the observed error rate exceeds a predefined threshold, the protocol is aborted. Otherwise, the process proceeds to the next step.

{\it Step 5}: For the remaining unused group of Cat states, each voter randomly selects one copy and performs a computational-basis measurement on their respective particle. Due to the intrinsic property of the Cat state, each voter obtains a unique measurement outcome, which serves as their private index numbers. For instance, suppose the particles distributed to the voters originate from the same Cat state: $|\Phi(0,1,2,\dots,n-1)\rangle=\frac{1}{\sqrt n}\sum^{n-1}_{l=0}|l,l+1,\dots,l+n-1\rangle$. When each voter measures their particle in the computational basis, the outcomes are $\{ |l\rangle, |l+1\rangle,\dots,|l+n-1\rangle \}$, which consists of $n$ distinct values. Thus, for each candidate $C_k$, every voter receives a unique index number, denoted as $N^k_0, N^k_1, \dots, N^k_{n-1}$, corresponding respectively to voters $V_0, V_1, \dots, V_{n-1}$.

\begin{figure*}[!htp]
\centering\includegraphics[width=6.2in]{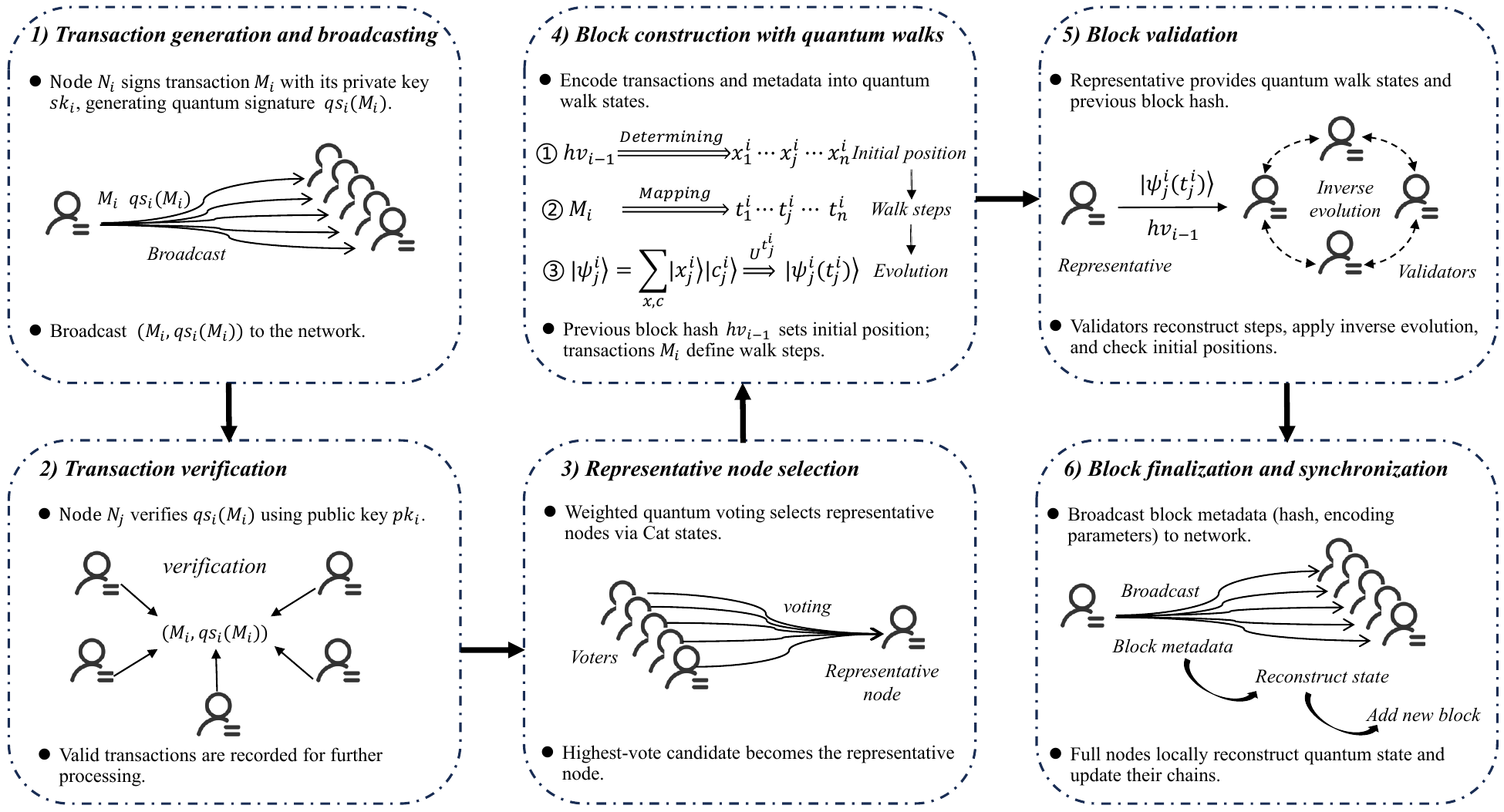}
\caption{ Workflow of the proposed quantum blockchain. }
\label{fig:5}       
\end{figure*}

\subsubsection{Secure quantum vote aggregation}
In this phase, each voter encodes their voting choice based on the previously assigned private index, and the candidate performs the final vote aggregation.  The detailed procedure is as follows.

{\it Step 1}: Similar to the index distribution phase, candidate $C_k$ prepares $n+\delta$ group Cat states. However, unlike in the previous phase, where each voter obtained a unique index, here the goal is to construct a quantum ballot box that accommodates weighted voting. To this end, each Cat state is initialized in the following form:
\begin{equation}
\begin{aligned}
|\Phi^\prime(0,0,\dots,0)\rangle  =\frac{1}{\sqrt d}\sum^{d-1}_{l^\prime=0}|l^\prime,l^\prime,\dots,l^\prime\rangle,
\end{aligned}
\end{equation}
where $d$  is the dimension of each qudit. To accommodate all possible voting weights and ensure correctness under modular arithmetic, the dimension $d$ must satisfy $d \ge T_v$. %This adaptation ensures that each voter's weighted vote can be correctly encoded and aggregated modulo $d$.

{\it Step 2}: $C_k$ divides each prepared Cat state into $n$ sequences and distributes them to the $n$ voters. To defend against potential eavesdropping attacks, decoy states are randomly inserted into each distribution sequence. As in the privacy index distribution phase, voters first perform channel security checks using the decoy states and then randomly select $\delta$ groups for state verification to ensure that the candidate has correctly prepared the states according to the protocol.

{\it Step 3}: For the remaining $n$ verified states, each voter selects one copy from each group and performs a Fourier-basis measurement on the received qudit. As a result, each voter obtains a sequence of $n$ measurement outcomes. These results can be organized into a matrix that represents a distributed quantum ballot box shared among all voters:
\begin{equation}
\scalebox{0.85}{
$\begin{pmatrix}
r^k_{0,0} &\dots  &  r^k_{0,l} &\dots & r^k_{0,n-1}\\
\vdots &\vdots  &  \vdots  &\vdots & \vdots\\
r^k_{g,0} &\dots  &  r^k_{g,l} &\dots & r^k_{g,n-1}\\
\vdots &\vdots  &  \vdots  &\vdots & \vdots\\
r^k_{n-1,0} &\dots  &  r^k_{n-1,l} &\dots & r^k_{n-1,n-1}
\end{pmatrix}$
},
\end{equation}
where $r^k_{g,l}\in\{0,1,\dots,d-1\}$ and $g=0,1,\dots,n-1$. Note that the $l$-th column $(r^k_{0,l}, r^k_{1,l}, \dots, r^k_{n-1,l})$ represents the outcomes obtained by voter $V_l$ from each of the $n$ Cat states, while the $g$-th row $(r^k_{g,0}, r^k_{g,1}, \dots, r^k_{g,n-1})$ corresponds to the measurement results for the $g$-th Cat state across all voters. According to the properties of the Cat state under Fourier-basis measurement, the outcomes from each row satisfy the condition:
\begin{equation}
\sum^{n-1}_{l=0} r^k_{g,l} \ \text{mod}\ d=0.
\end{equation}

{\it Step 4}: Then, each voter casts their vote by updating the ballot box according to their previously obtained private index $N^k_0, N^k_1, \dots, N^k_{n-1}$. Specifically, voter $V_l$ encodes their vote $v^k_l\in\{0,1,\dots,\tilde{w}_l\}$ by modifying the value of $(r^k_{0,l}, r^k_{1,l}, \dots, r^k_{n-1,l})$. The updated rule as follows:
\begin{equation}
r^{\prime k}_{g,l}=\left\{
\begin{aligned}
& r^k_{g,l}+v^k_l \ \text{mod}\ d      \ & \text{if} \ g=N^k_l\\
&  r^k_{g,l}\    &  \text{if} \ g\neq N^k_l
\end{aligned}
\right..
\end{equation}

{\it Step 5}:  After the voting concludes, all voters simultaneously publish their updated ballot values. These values are then used to reconstruct the updated ballot matrix for candidate $C_k$:
\begin{equation}
\scalebox{0.85}{
$\begin{pmatrix}
r^{\prime k}_{0,0} &\dots  &  r^{\prime k}_{0,l} &\dots & r^{\prime k}_{0,n-1} &result^k_{0}\\
\vdots &\vdots  &  \vdots  &\vdots &  &\vdots\\
r^{\prime k}_{g,0} &\dots  &  r^{\prime k}_{g,l} &\dots & r^{\prime k}_{g,n-1} &result^k_g\\
\vdots &\vdots  &  \vdots  &\vdots &  &\vdots\\
r^{\prime k}_{n-1,0} &\dots  &  r^{\prime k}_{n-1,l} &\dots & r^{\prime k}_{n-1,n-1} &result^k_{n-1}
\end{pmatrix}$
},
\end{equation}
where $result^k_g=\sum^{n-1}_{l=0} r^{\prime k}_{g,l}\ \text{mod}\ d$. According to formula (15), it can be easily deduced that
\begin{equation}
result^k_g=\sum^{n-1}_{l=0} r^{\prime k}_{g,l}\ \text{mod}\ d= \sum^{n-1}_{l=0} r^k_{g,l}+v^k_l  \ \text{mod}\ d= v^k_l .
\end{equation}
By this stage, every voter has successfully cast their vote for each candidate $C_k$.

\subsubsection{Representative node selection}
Each candidate $C_k$ computes their total number of votes by summing the individual row results in their vote matrix:
$
result^k = \sum_{g=0}^{n-1} result^k_g.
$
The top $r$ candidates with the highest totals are selected as representative nodes. To ensure correctness, each voter can verify whether their own vote has been accurately counted. The voting weight of voter $V_l$ is $\tilde{w}_l$. It is assumed that all of the votes are cast for the candidate at index $l$. Then $V_l$ checks:
\begin{equation}
\sum_{k=1}^{m} \text{result}^k_l = \tilde{w}_l.
\end{equation}
If the equation holds, $V_l$ can be confident that his vote has been correctly included in the final tally.

\subsubsection{Block production and validation}
 After the representative nodes are elected, they are arranged in a random order to sequentially perform block production. Let $f_t$ be the fixed time interval for block generation. When it is a representative node’s turn, the node collects all valid transactions that occurred during the time window $f_t$. These transactions are verified using quantum digital signature techniques~\cite{44,45,46} to ensure their authenticity and integrity. Once verified, the transactions are packaged into a new block (see Section 3.1 for details on transaction data encoding). If the representative node fails to produce a block within the designated time, the block is considered invalid, and the transactions are forwarded to the next node in the sequence.

To ensure the correctness of block production, designated validator nodes are responsible for subsequent verification. For the $i$-th new block, the representative node is required to send its quantum walk states ${|\psi^i_1(t^i_1)\rangle, |\psi^i_2(t^i_2)\rangle, \dots, |\psi^i_n(t^i_n)\rangle}$ along with the hash value $hv_{i-1}$ of the previous block. Each validator determines the quantum walk steps $\{t^i_1, t^i_2,\dots,t^i_n\}$ based on the transaction data and timestamp of the new block, and applies the inverse evolution operator $U^{-t^i_j}$ to the received quantum states, thereby returning the quantum walker to its initial position. By measuring and extracting the position information $x^i_j$, the validator compares it with the expected initial position encoded by $hv_{i-1}$. If the measurement result matches the expected value, the block is considered valid.

\subsubsection{Incentive and re-election mechanism}

To motivate participation and honest behavior, nodes are rewarded for successful voting and block outs. These rewards can be issued through blockchain transactions, usually in the form of native tokens or virtual currencies.  If a representative node behaves dishonestly or fails to perform its duties, it may be penalized or removed. After a complete round of block production, a new round of quantum voting can be initiated to re-elect representative nodes. Honest nodes may be re-elected, while malicious or idle nodes are excluded from future rounds.

 \subsection{ Workflow of the proposed quantum blockchain}
The previous two sections have outlined the structure and consensus mechanism of the quantum blockchain. This section further provides an overview of the overall workflow of the proposed framework, as shown in Fig. 5. The workflow leverages quantum walks to construct the blockchain architecture and establish connections between blocks, thereby ensuring block integrity. Meanwhile, consensus is achieved through a weighted quantum voting-driven mechanism. A brief description of each phase is presented below.

\subsubsection{Transaction generation and broadcasting} Node $N_i$ uses its private key $sk_i$ to generate a quantum signature $qs_i(M_i)$ for the transaction $M_i$. Subsequently, it broadcasts the tuple $(M_i, qs_i(M_i))$, which includes both the transaction and its signature, to the entire blockchain network.

\subsubsection{Transaction verification}  Upon receiving the transaction, any node $N_j$ utilizes the public key $pk_i$ of node $N_i$ to  verify the signature $qs_i(M_i)$, thereby assessing the validity and legitimacy of the transaction. If the verification passes, node $N_j$ records the transaction for further processing.

\subsubsection{Representative node selection}
To reach consensus, a weighted quantum voting-based delegated proof of stake mechanism is employed to select representative nodes. Leveraging the properties of Cat states, each voter can securely cast their vote according to their respective weight, while candidates are able to aggregate the votes without revealing individual voting choices. The candidate with the highest number of votes is elected as the representative node and is responsible for generating the new block in the next round.

\subsubsection{Block construction with quantum walks} 
Once representatives are elected, they are ordered randomly to produce blocks in sequence. For the $i$-th block, during its assigned time slot, the representative node collects validated transactions $M_i$ and encodes them, along with associated metadata, into the discrete-time quantum walk framework. Specifically, the hash of the previous block $hv_{i-1}$ determines the initial positions of the walk states, i.e.,
\begin{equation}
hv_{i-1} \;\mapsto\; \{\,|x^i_1\rangle,|x^i_2\rangle,\dots,|x^i_n\rangle \,\},
\end{equation}
where $|x^i_j\rangle$ denotes the $j$-th position state. Meanwhile, the verified transactions $M_i$ are partitioned into $n$ components,
\begin{equation}
M_i \;\mapsto\; \{\,t^i_1, t^i_2, \dots, t^i_n\,\},
\end{equation}
as defined in Eq.~(10). During the quantum walk evolution, the states $\{|\psi^i_j\rangle\}$ are updated by embedding the corresponding transaction components, resulting in
\begin{equation}
|\psi^i_j\rangle \;\xrightarrow{U(t^{i}_{j})}\; |\psi^i_j(t^i_j)\rangle, \quad j=1,2,\dots,n.
\end{equation}
The final set of evolved states constitutes the candidate block, inherently binding the new block to both the past history (through $hv_{i-1}$) and the current data (through $M_i$).

%$\{\,|\psi^i_1(t^i_1)\rangle, |\psi^i_2(t^i_2)\rangle, \dots, |\psi^i_n(t^i_n)\rangle \,\}$

\subsubsection{Block validation}  The representative node provides the quantum walk states and the previous block’s hash to validators. Each validator reconstructs the expected quantum walk steps from the announced transactions and metadata, then applies inverse evolution to test whether the received states collapse to the correct initial positions. If the reconstructed results match, the block is deemed valid. Validators then aggregate their results using a quantum summation protocol, and if a sufficient majority (e.g., two-thirds) approve, the block is accepted. Otherwise, the block is discarded.

\subsubsection{Block finalization and synchronization}
Once validated, the block's classical metadata, such as the block hash and encoding parameters, is broadcast to the network. With transaction data and signatures already verified, full nodes locally reconstruct the quantum state of the new block and update their chains, ensuring system-wide consistency. %Quantum walk dynamics structurally link blocks, replacing entanglement-based schemes and improving scalability.

\section{Security analysis and simulation}

\subsection{Security analysis} 
In this section, we analyze the security properties of the proposed quantum blockchain, focusing on its resilience against both classical and quantum adversaries. Our goal is to demonstrate that the system guarantees (i) correctness of transactions, (ii) immutability of blocks, and (iii) fairness of consensus, even in the presence of malicious participants. These properties are analyzed with respect to the capabilities of potential adversaries, which we formally define below. 

{\it Adversarial model:} 
We assume a quantum polynomial-time (QPT) adversary $\mathcal{A}$ with the following capabilities:
\begin{itemize}
    \item $\mathcal{A}$ may attempt to forge transactions by producing counterfeit signatures.  
    \item $\mathcal{A}$ may tamper with blocks by modifying encoded quantum walk states or replaying past states.  
    \item $\mathcal{A}$ may bias consensus by manipulating votes or producing conflicting chains.  
\end{itemize}
The adversary is bounded by the principles of quantum mechanics, including the no-cloning theorem and measurement disturbance.

\subsubsection{Transaction correctness and authenticity}

In our scheme, the correctness and authenticity of transactions are guaranteed by the quantum digital signature mechanism, which ensures that each transaction is generated by the legitimate sender and the content has not been tampered with.

\begin{definition}%[Transaction Correctness]
A transaction ledger is considered correct if every transaction $M_i$ recorded in it has been validated according to the protocol rules, which requires that $M_i$ was generated and signed by the legitimate owner of the originating account.
\end{definition}

\begin{theorem}
Let node $N_i$ use a quantum signature scheme. For any QPT adversary $\mathcal{A}$, the probability that $\mathcal{A}$ produces a fraudulent tuple $(M', qs_i(M'))$ that passes verification under public key $pk_i$ is negligible in the security parameter $\kappa$.
\end{theorem}

\begin{proof} 
Each transaction is signed by a quantum digital signature scheme \cite{44} instantiated from a quantum one-way function $f: \{0,1\}^L \to (\mathbb{C}^2)^{\otimes H}$, which maps an $L$-bit string to an $H$-qubit quantum state. The private key of node $N_i$ is $sk_i = (k_{i,0}, k_{i,1})$, where each $k_{i,j}$ is sampled uniformly at random from $\{0,1\}^L$. The corresponding public key $pk_i = \{ |f_{k_{i,0}}\rangle, |f_{k_{i,1}}\rangle \}$ consists of quantum states that are distributed to other nodes for signature verification.

By Holevo’s theorem, if an adversary obtains at most $n$ copies of the public key states, they can extract at most $nH$ classical bits of information from each quantum state. Since the private key length is $L$ bits, the remaining uncertainty is $L-nH$ bits per string. Thus, the adversary’s success probability in recovering the full private key is bounded by:
\begin{equation}
\Pr[\mathcal{A} \text{ forges a valid signature}] \leq 2^{-(L-nH)}\cdot (2b),
\end{equation}
where $b$ is the total number of single-bit messages for which the adversary attempts to create forgeries. For sufficiently large $L \gg nH$, this probability is negligible in $\kappa$. Hence, the adversary cannot generate a valid quantum signature without knowing $sk_i$, and forging a transaction is infeasible.

This ensures that all accepted transactions are authentically generated by the claimed sender $N_i$. Transaction forgery, replay, or double-spending attempts by adversaries will be rejected except with negligible probability, guaranteeing correctness and integrity of the transaction ledger.
\end{proof} 

\subsubsection{Immutability of quantum blockchain} 
The immutability of our blockchain is underpinned by quantum walk evolution. Each walk originates from an initial position determined by the predecessor's hash and evolves according to the block's transaction data. This mechanism binds a block’s internal content to its position, such that any historical alteration results in a cascading validation failure, thereby ensuring immutability.

\begin{definition}
A blockchain $C = (B_1, B_2, \dots, B_N)$ is considered valid if, for every block $B_i$, it satisfies two conditions:

(1) The initial quantum walk positions ${|x^i_j\rangle}$ used in $B_i$ are correctly derived from the hash of the preceding block, i.e., $hv_{i-1} = \text{QFHs}(B_{i-1})$.

(2) The final quantum states ${|\psi^i_j(t^i_j)\rangle}$ in $B_i$ pass the backward evolution test using the step counts ${t^i_j}$ derived from the block’s transactions $M_i$. Concretely, for all $j=1,2,\dots,n$, it must hold that $U^{-t^i_j} |\psi^i_j(t^i_j)\rangle \mapsto |x^i_j\rangle$,
ensuring each state recovers its original position.
\end{definition}

\begin{theorem}
Let $C = (B_1, \dots, B_N)$ be a valid blockchain. For any QPT adversary $\mathcal{A}$, the probability of successfully modifying a confirmed block $B_k$ ($k < N$) into $B'_k$ such that the modified chain $C'=(B_1,\dots,B'_k,\dots,B_N)$ is accepted as valid by an honest verifier is negligible in the security parameter $\lambda$.
\end{theorem}

\begin{proof}
Suppose an adversary $\mathcal{A}$ attempts to replace the original block $B_k$ with a fraudulent block $B^\prime_k$ containing modified transactions $M^\prime_k \neq M_k$. For the altered chain $C'$ to be valid, the adversary must simultaneously pass the following two independent checks:

{\it Internal consistency check.}  
The verifier computes the initial positions ${|x^k_j\rangle}$ from the hash of the previous block, $hv_{k-1} = \text{QFHs}(B_{k-1})$. They then perform the backward evolution test on $B'_k$ with its states ${|\psi'^k_j\rangle}$ and step counts ${t'^k_j}$. Correctness requires that
\begin{equation}
U^{-t'^k_j} |\psi'^k_j\rangle \mapsto |x^k_j\rangle, \quad \forall j.
\end{equation}
Since $M'_k \neq M_k$, there exists at least one index $j$ for which $t'^k_j \neq t^k_j$, and hence the backward evolution will fail to reproduce the original initial state $|x^k_j\rangle$. Due to the non-trivial dynamics of quantum walks, it is computationally infeasible for an adversary to construct alternative states $|\psi'^k_j\rangle$ and step counts $t'^k_j$ that collectively yield the prescribed initial positions. Consequently, $\mathcal{A}$ can pass the internal consistency check only with negligible probability.

{\it External linkage check.}  
Even if $\mathcal{A}$ miraculously succeeds in forging $B'_k$ that passes the internal check, the block hash values will differ, since
\begin{equation}
hv^\prime_k = \text{QFHs}(B^\prime_k) \neq hv_k = \text{QFHs}(B_k).
\end{equation}
The subsequent block $B_{k+1}$ was generated using $hv_k$ to define its initial positions ${|x^{k+1}_j\rangle}$. A verifier inspecting the link between $B'_k$ and $B_{k+1}$ would instead use $hv^\prime_k$, which does not match the actual initial positions in $B_{k+1}$. Thus, the linkage validation fails. To restore validity, the adversary would need to recompute not only $B_{k+1}$ but all subsequent blocks, effectively reconstructing the entire chain from $B_{k+1}$ to $B_N$. This task is infeasible without overwhelming the cumulative computational power of the honest network.

Since both internal consistency and external linkage must hold simultaneously, and each can be broken only with negligible probability, we conclude that
\begin{equation}
\Pr[\mathcal{A} \text{ forges a valid chain}] \leq \text{negl}(\lambda).
\end{equation}
Therefore, the immutability of this blockchain ledger is cryptographically guaranteed.
\end{proof}

\subsubsection{Fairness of the proposed QDPoS consensus} 
Our proposed QDPoS consensus protocol ensures fairness through a weighted quantum voting scheme,  where voters encode their choices into measurement results derived from Cat states and candidates aggregate votes without learning individual preferences.  Even an adversary  $\mathcal{A}$ with unbounded quantum power cannot amplify its influence beyond its legitimately assigned voting weight.

\begin{definition}
A consensus protocol is fair if the influence of any participant (including an adversary) on the final decision is strictly determined by its allocated voting weight, and cannot be amplified through computational power or external capabilities.
\end{definition}

\begin{theorem}
For any QPT adversary $\mathcal{A}$, its computational advantages, including access to a quantum computer, cannot be used to bias the outcome or compromise the fairness of the QDPoS consensus mechanism.
\end{theorem}

\begin{proof}
We prove this theorem by analyzing an adversary $\mathcal{A}$'s potential attack  across the critical phases. We will show that each potential attack is rendered ineffective.

{\it Resistance to vote forgery and amplification.} An adversary's most direct path to subverting fairness is to cast more votes than their stake allows.
In our voting protocol, each voter $V_l$ is assigned a fixed voting weight $\tilde{w}_l$, and receives a unique index $N^k_l$ from candidate $C_k$. After that, $V_l$ casts $v^k_l \in \{0,1,\dots,\tilde{w}_l\}$ votes at the position indexed by $N^k_l$. Candidate $C_k$ then aggregates the submitted vectors by computing their component-wise sum. Owing to the uniqueness of indices, the entry corresponding to $N^k_l$ in the aggregated result encodes the contribution of voter $V_l$. Finally, the candidates announce their respective statistical results, and $V_l$ checks whether $\sum^m_{k=1} result^k_l=\tilde{w}_l$. If the condition holds, $V_l$ can be assured that their ballot has been faithfully included in the final tally, thereby preventing any adversary from forging additional votes or amplifying their influence.

{\it  Resistance to eavesdropping and vote tampering.} A typical attack by an adversary $\mathcal{A}$ involves the interception and tampering of votes from honest participants during transmission over the quantum channel, which may occur either in the private index distribution phase or the secure vote aggregation phase. To counter such attacks, our protocol incorporates decoy states into the quantum state transmission. Any attempt to perform a non-orthogonal measurement will perturb these quantum states with a high probability, leading to the detection of the attack during the verification step. Furthermore, the protocol leverages the properties of Cat states to guarantee vote security. Specifically, measurements on Cat states under different bases exhibit deterministic algebraic correlations. Any tampering, no matter how sophisticated, will disrupt this delicate correlational structure and be therefore detectable.

The vote tallies are secured against both forgery and tampering. As a result, the representative selection, which relies on these tallies through a deterministic public computation, is also protected from manipulation. Throughout all phases, the influence of an adversary $\mathcal{A}$ remains strictly confined to its legitimate stake. The probability that $\mathcal{A}$ can amplify its influence beyond its allocated weight is negligible. Hence, the fairness of the proposed QDPoS consensus is guaranteed.
\end{proof}

\subsection{Simulation} 
%Here, we give an example of the proposed quantum blockchain scheme.
To validate the feasibility of our proposed quantum blockchain framework, we conducted circuit simulations using IBM Qiskit, focusing on two core quantum components: (1) the weighted quantum voting consensus, and (2) the quantum walk-based block construction.

\subsubsection{Simulation of weighted quantum voting consensus}  This section focuses on validating the weighted quantum voting  protocol, which is the core of our consensus mechanism. To illustrate the process, we consider a simplified scenario with four voters ${V_0, V_1, V_2, V_3}$ participating in an election to select one representative from two candidates ${C_1, C_2}$. Each voter’s voting power is proportional to its stake. For example, assigning weights $\{0.3, 0.3, 0.2, 0.2\}$ to voters ${V_0, V_1, V_2, V_3}$ and normalizing to 10 discrete vote units yields allocations of 3, 3, 2, and 2 votes, respectively.
\begin{figure}[h]
    \centering
    \includegraphics[width=3.2in]{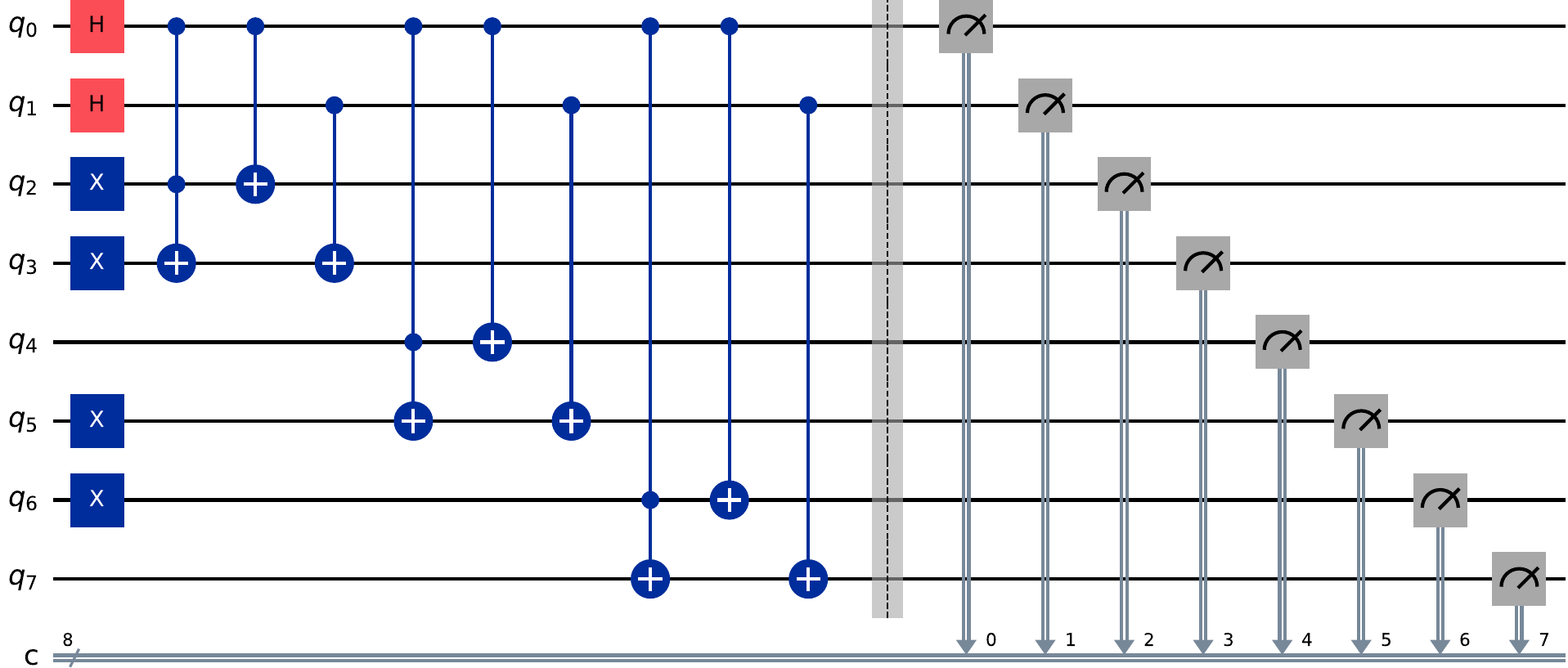}
    \caption{Schematic of the privacy index distribution circuit. The registers ($q_0,q_1,\dots,q_7$) are grouped in pairs, with each pair representing a 4-level particle.}%, which is then sent to voters ${V_0, V_1, V_2, V_3}$. }
\end{figure}

According to the proposed scheme and simulation setup, candidates $C_1$ and $C_2$ are required to prepare 4-level 4-particle Cat states and distribute them to four voters ${V_0, V_1, V_2, V_3}$, thereby establishing the corresponding voting privacy indices. Subsequently, $C_1$ and $C_2$ further prepare Cat states for the aggregation of the final voting information. The corresponding quantum circuits are shown in Fig. 6 and Fig. 7, where Fig. 6 illustrates the distribution of privacy indices and Fig. 7 presents the aggregation of voting privacy. Specifically, we assume that the 4-level 4-particle Cat state used for distribution is 
\begin{equation}
|\Phi(0,3,2,1)\rangle=\frac{1}{\sqrt 4}\sum^{3}_{l=0}|l,l+3,l+2,l+1\rangle, 
\end{equation}
where the measurement outcomes of voters ${V_0, V_1, V_2, V_3}$ on their respective particles determine their privacy indices. For the aggregation phase, we assume that the 4-level 4-particle Cat state is initialized as
\begin{equation}
|\Phi^\prime(0,0,0,0)\rangle=\frac{1}{\sqrt 4}\sum^{4}_{l^\prime=0}|l^\prime,l^\prime,l^\prime,l^\prime \rangle, 
\end{equation}
upon which voters ${V_0, V_1, V_2, V_3}$ perform Fourier-based measurements to jointly encode their voting information.
\begin{figure}[h]
    \centering
    \includegraphics[width=3.2in]{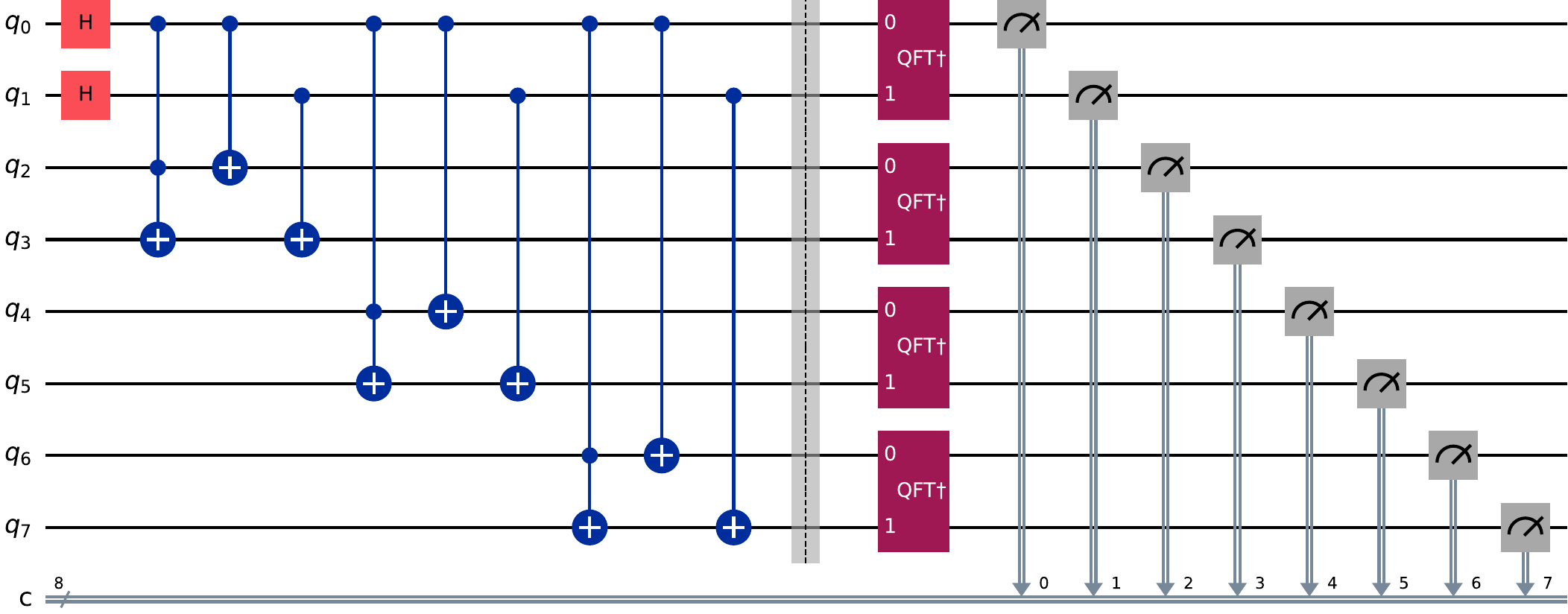}
    \caption{Quantum circuit for voting aggregation, where each voter's Fourier-basis measurements collectively satisfy a modulo‑4 constraint of zero. }
\end{figure}

The simulation results of these processes are presented in Fig. 8, including the obtained privacy indices of each voter and the corresponding ballot matrix derived from the Fourier-based measurements.

\begin{figure}[h]
\centering
\begin{subfigure}{0.48\columnwidth}
\centering
\includegraphics[width=\linewidth,height=1.2in]{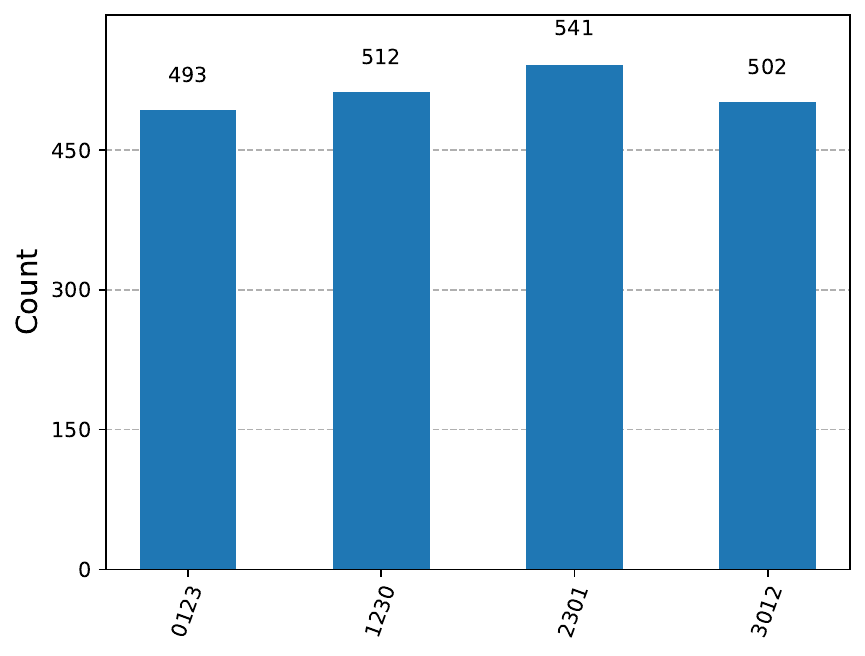}
\caption{ }
\label{fig:fig1}
\end{subfigure}
\hfill
\begin{subfigure}{0.48\columnwidth}
\centering
\includegraphics[width=\linewidth, height=1.2in]{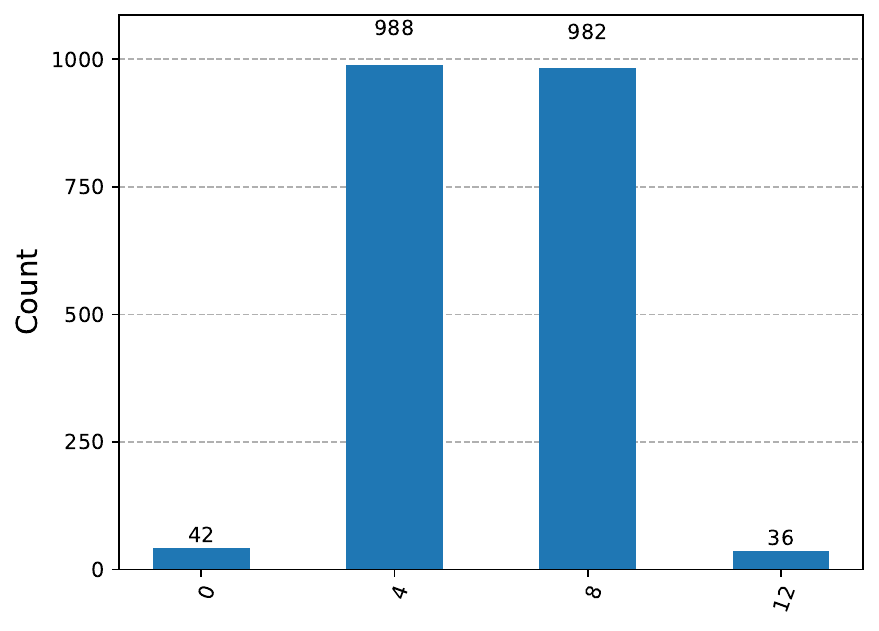}
\caption{ }
\label{fig:fig2}
\end{subfigure}
\caption{Simulation results corresponding to Figs. 6 and 7. Fig. 8(a) shows each voter's privacy index, and Fig. 8(b) shows the sum of the voters' Fourier-basis measurements.}
\end{figure}

Following the circuit constructions, we now illustrate a concrete example of the voting process.
Assume that the four voters cast their votes for candidates $C_1$ and $C_2$ as
\begin{equation}
\begin{aligned}
&(v^1_0,v^1_1,v^1_2,v^1_3)=(2,1,0,1), \\
& (v^2_0,v^2_1,v^2_2,v^2_3)=(1,2,2,1).
\end{aligned}
\end{equation}
The distributed privacy indices are set as $\{1,2,3,0\}$ for $C_1$ and $\{3,0,1,2\}$ for $C_2$. During the aggregation phase, the ballot matrices $r^1_{g,l}$ and $r^2_{g,l}$ are obtained as
\begin{equation}
\small
r^1_{g,l}=
\begin{pmatrix}
0 & 2  &  0 & 2 \\
3 & 0  &  2 & 3 \\
1 & 3  &  1 & 3 \\
2 & 0  &  2 & 0 
\end{pmatrix},
\quad 
r^2_{g,l}=
\begin{pmatrix}
1 & 1  &  1 & 1 \\
2 & 3  &  0 & 3 \\
1 & 3  &  2 & 2 \\
2 & 3  &  1 & 2 
\end{pmatrix},
\end{equation}
where each row sums to 0 modulo 4. The overall voting procedure is summarized in Table I, where the final tallies are publicly announced, enabling each voter to verify that their ballot has been faithfully included. 
\begin{table}[htbp]
\centering
\footnotesize \caption{\small Example of the voting process with privacy indices}
\begin{tabular}{cccccc}
\toprule
\textbf{ } & \textbf{$V_0$} & \textbf{$V_1$} & \textbf{$V_2$} & \textbf{$V_3$} & \textbf{Result} \\
\midrule
$C_1$ & 2 & 1 & 0 & 1 & \textbf{4} \\
\noalign{\vspace{2pt}} 
$r^{\prime1}_{0,l}$ & 0 & 2 & 0 & 2+1 & 1 \\
\noalign{\vspace{2pt}} 
$r^{\prime1}_{1,l}$ & 3+2 & 0 & 2 & 3 & 2 \\
\noalign{\vspace{2pt}} 
$r^{\prime1}_{2,l}$ & 1 & 3+1 & 1 & 3 & 1 \\
\noalign{\vspace{2pt}} 
$r^{\prime1}_{3,l}$ & 2 & 0 & 2+0 & 0 & 0 \\
\midrule
$C_2$ & 1 & 2 & 2 & 1 & \textbf{6} \\
\noalign{\vspace{2pt}} 
$r^{\prime2}_{0,l}$ & 1 & 1+2 & 1 & 1 & 2 \\
\noalign{\vspace{2pt}} 
$r^{\prime2}_{1,l}$ & 2 & 3 & 0+2 & 3 & 2 \\
\noalign{\vspace{2pt}} 
$r^{\prime2}_{2,l}$ & 1 & 3 & 2 & 2+1 & 1 \\
\noalign{\vspace{2pt}} 
$r^{\prime2}_{3,l}$ & 2+1 & 3 & 1 & 2 & 1 \\
\bottomrule
\end{tabular}%
\end{table}

%\begin{table}[htbp]
%\centering
%\footnotesize \caption{\small Example of the voting process with privacy indices}
%\begin{tabular}{cccccc|cccccc}
%\toprule
%\textbf{$C_1$ } & \textbf{$V_0$} & \textbf{$V_1$} & \textbf{$V_2$} & \textbf{$V_3$} & \textbf{Result}& \textbf{$C_2$ } & \textbf{$V_0$} & \textbf{$V_1$} & \textbf{$V_2$} & \textbf{$V_3$} & \textbf{Result} \\
%\midrule
%$r^{\prime1}_{0,l}$ & 0 & 2 & 0 & 2+1 & 1  & $r^{\prime2}_{0,l}$ & 1 & 1+2 & 1 & 1 & 2\\
%$r^{\prime1}_{1,l}$ & 3+2 & 0 & 2 & 3 & 2  & $r^{\prime2}_{1,l}$ & 2 & 3 & 0+2 & 3 & 2\\
%$r^{\prime1}_{2,l}$ & 1 & 3+1 & 1 & 3 & 1  & $r^{\prime2}_{1,l}$ & 2 & 3 & 0+2 & 3 & 2\\
%$r^{\prime1}_{3,l}$ & 2 & 0 & 2+0 & 0 & 0  & $r^{\prime2}_{3,l}$ & 2+1 & 3 & 1 & 2 & 1\\
%\bottomrule
%\end{tabular}%
%\end{table}

\begin{figure*}[!htbp]
\centering
\begin{subfigure}[b]{0.5\textwidth}
    \centering
    \includegraphics[width=\textwidth]{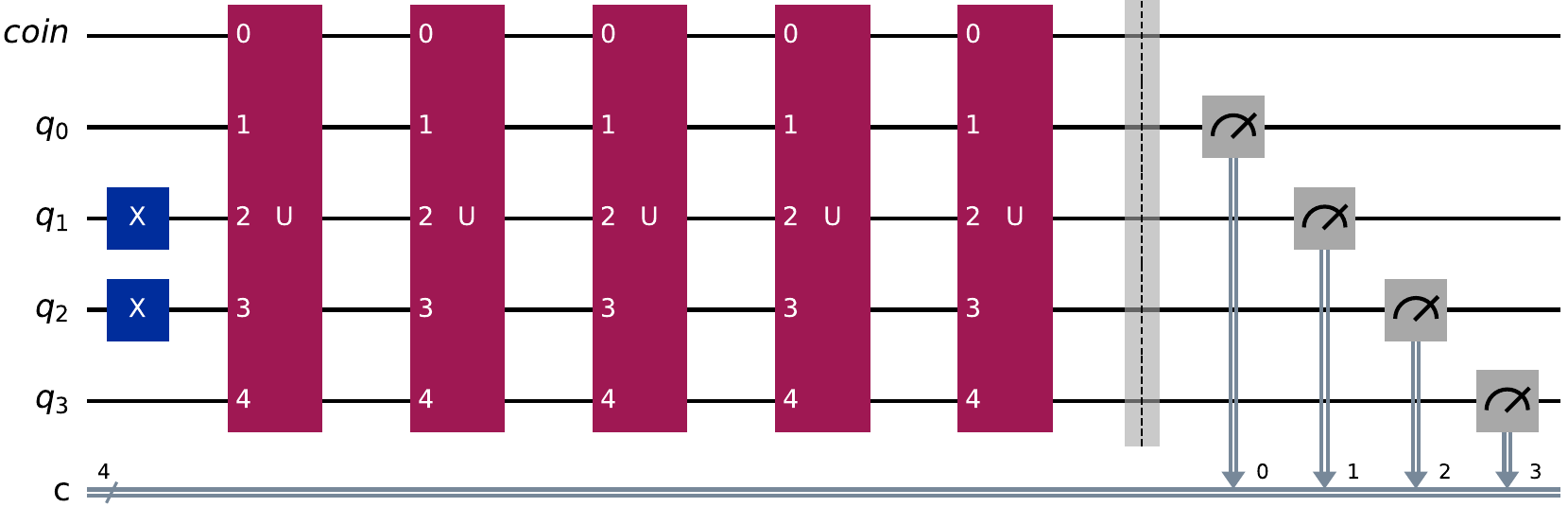}
    \caption{ }
    \label{fig:subfig1}
\end{subfigure}
\hfill
\begin{subfigure}[b]{0.45\textwidth}
    \centering
    \includegraphics[width=\textwidth]{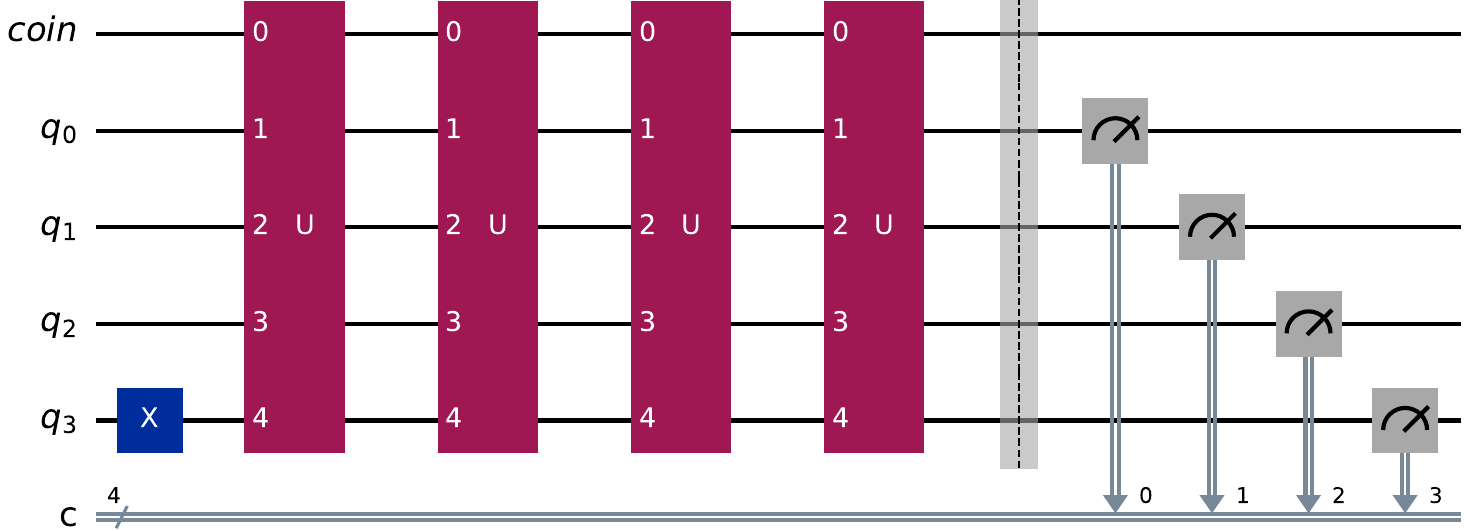}
    \caption{ }
    \label{fig:subfig2}
\end{subfigure}

\vspace{0.2cm}

\begin{subfigure}[b]{0.24\textwidth}
    \centering
    \includegraphics[width=\textwidth]{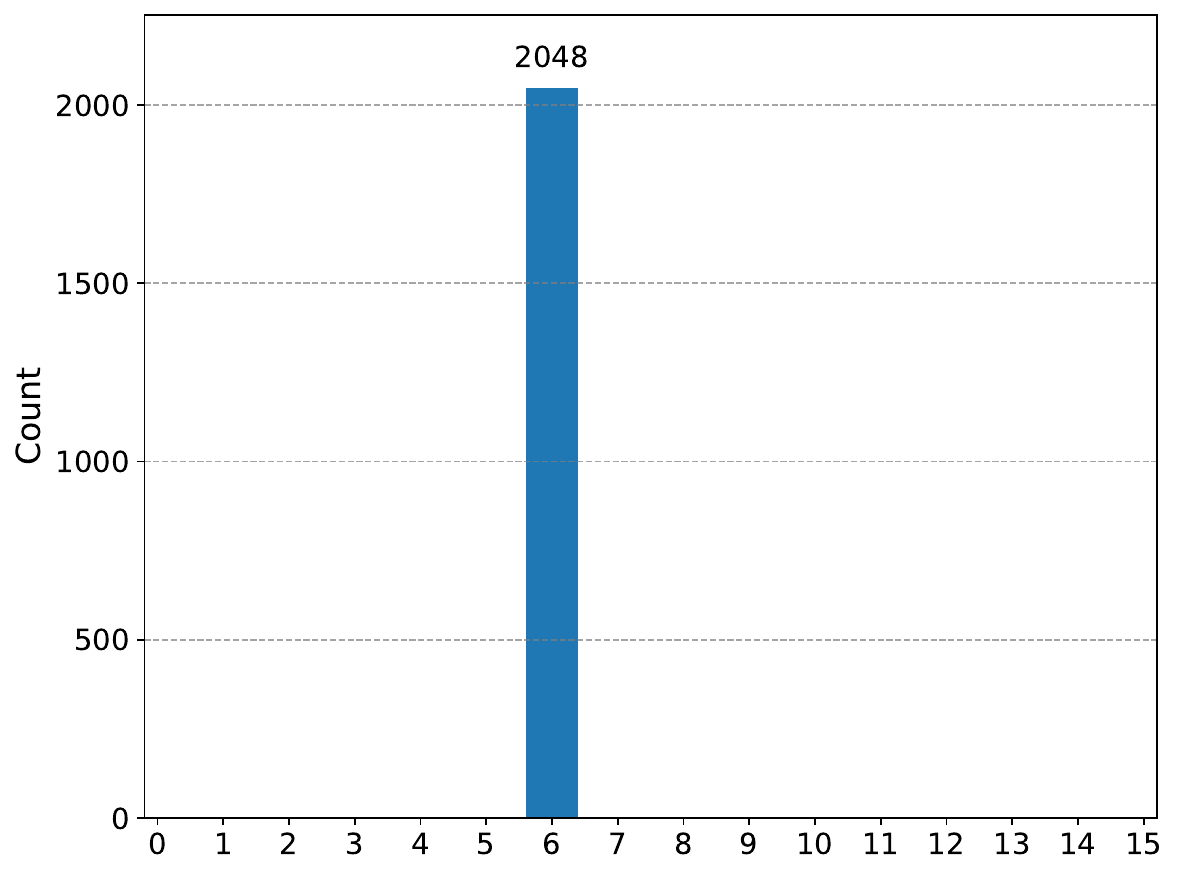}
    \caption{ }
    \label{fig:subfig3}
\end{subfigure}
\hfill
\begin{subfigure}[b]{0.24\textwidth}
    \centering
    \includegraphics[width=\textwidth]{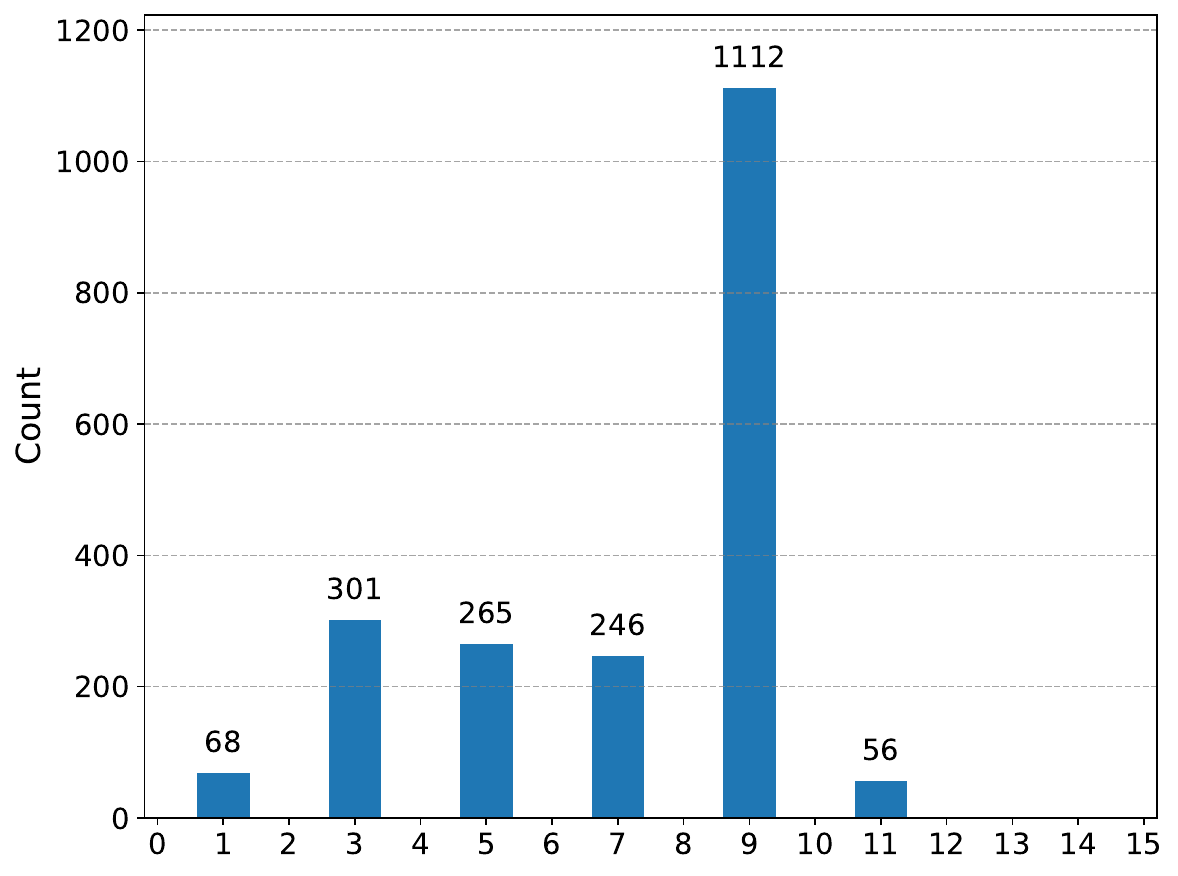}
    \caption{ }
    \label{fig:subfig4}
\end{subfigure}
\hfill
\begin{subfigure}[b]{0.24\textwidth}
    \centering
    \includegraphics[width=\textwidth]{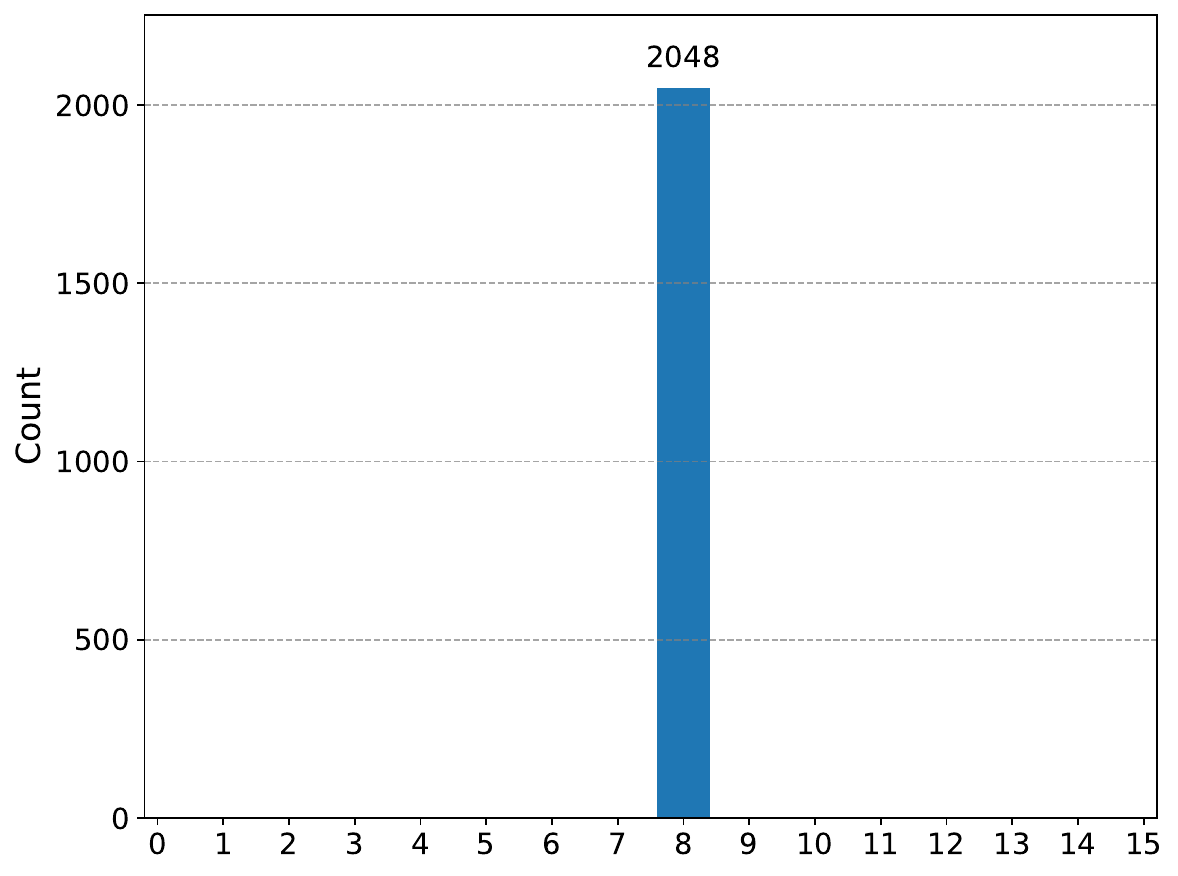}
    \caption{ }
    \label{fig:subfig5}
\end{subfigure}
\hfill
\begin{subfigure}[b]{0.24\textwidth}
    \centering
    \includegraphics[width=\textwidth]{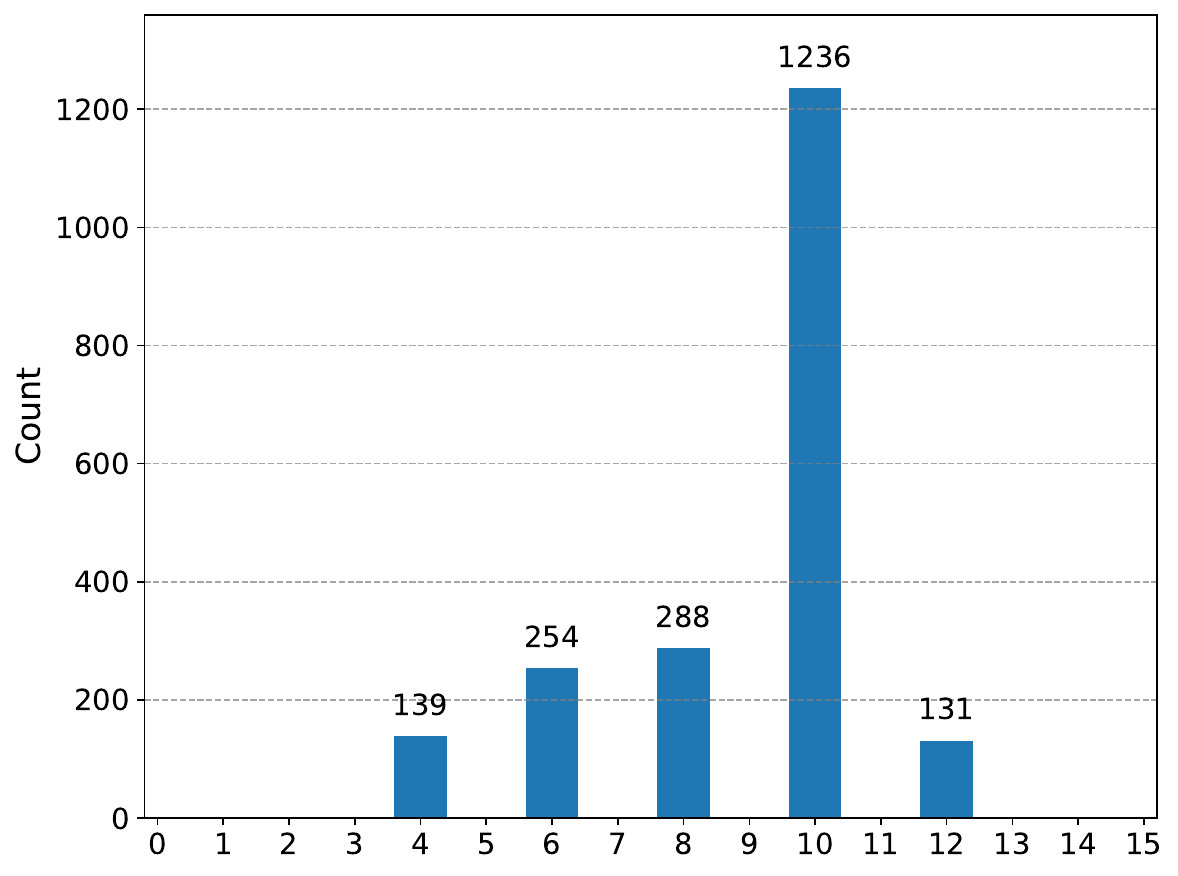}
    \caption{ }
    \label{fig:subfig6}
\end{subfigure}

\caption{Quantum walk circuits and simulation results for block construction. This figure illustrates the evolution of two initial states. The left panels depict the 5-step evolution ($U^5$) of a walker starting at position $|6\rangle$: (a) shows the implementing circuit, while (c) and (d) show the initial and final probability distributions. The right panels show the 4-step evolution ($U^4$) of a walker starting at $|8\rangle$: (b) is the corresponding circuit, and (e) and (f) are the distributions before and after the walk.}
\label{fig:main}
\end{figure*}

For instance, voter $V_1$ has privacy indices 2 and 0 with a total voting weight of 3. By directly inspecting Table I, $V_1$ can confirm that the aggregated tally of their votes is indeed 3, ensuring the verifiability of the process. In this case, candidate $C_2$ receives 6 votes, while candidate $C_1$ receives 4 votes, making $C_2$ the elected representative node.

\subsubsection{Simulation of quantum walk-based block construction} In this part, we simulate the construction of quantum blocks based on discrete-time quantum walks. Taking the $i$-th block as an example, the initial position state of the quantum walk is determined by the hash value $hv_{i-1}$ of the previous block, while the number of walk steps $t^i_j$ encodes the transactions of the current block.

To simplify the simulation, we utilize a 16-dimensional position space ${|0\rangle, \dots, |15\rangle}$, realized with a 4-qubit register. Then, the previous block's hash seeds two independent initial walker states, $|x^i_1\rangle \otimes |c_0\rangle$ and $|x^i_2\rangle \otimes |c_0\rangle$. For this scenario, we assume the hash maps to initial positions $|x^i_1\rangle = |6\rangle$ and $|x^i_2\rangle = |8\rangle$, while the coin is initialized to $|c_0\rangle = |0\rangle$. The current block's transactions are encoded into the dynamics of the walk. We consider a block with two transactions that dictate the number of evolution steps, $t^i_1 = 5$ and $t^i_2 = 4$. 

Thus, the construction of the $i$-th block can be simulated by applying the quantum walk operator $U$ for the specified number of steps, i.e., $U^{5}$ acting on $|6\rangle\otimes|0\rangle$ and $U^{4}$ acting on $|8\rangle\otimes|0\rangle$. The resulting block is represented by the pair of evolved states that encode both the linkage to the previous block’s hash and the transactions of the current block. The corresponding circuits and simulation results are shown in Fig. 9. The detailed circuit construction of the evolution operator $U$ and its inverse $U^\dagger$ is provided in Fig. 10.
\begin{figure}[h]
\centering
\begin{subfigure}{0.7\columnwidth}
\centering
\includegraphics[width=\linewidth]{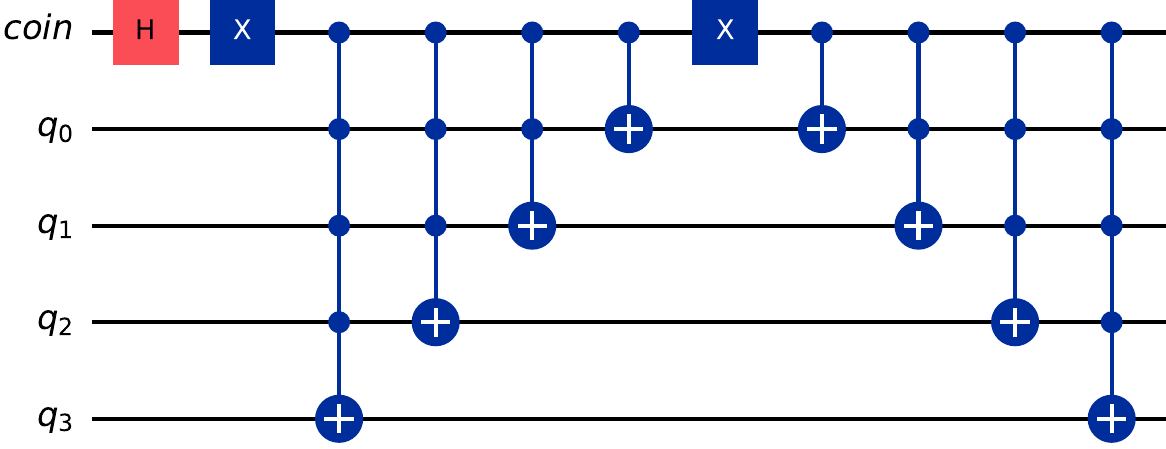}
\caption{ Circuit for $U$}
\label{fig:fig1}
\end{subfigure}
\hfill
\begin{subfigure}{0.7\columnwidth}
\centering
\includegraphics[width=\linewidth]{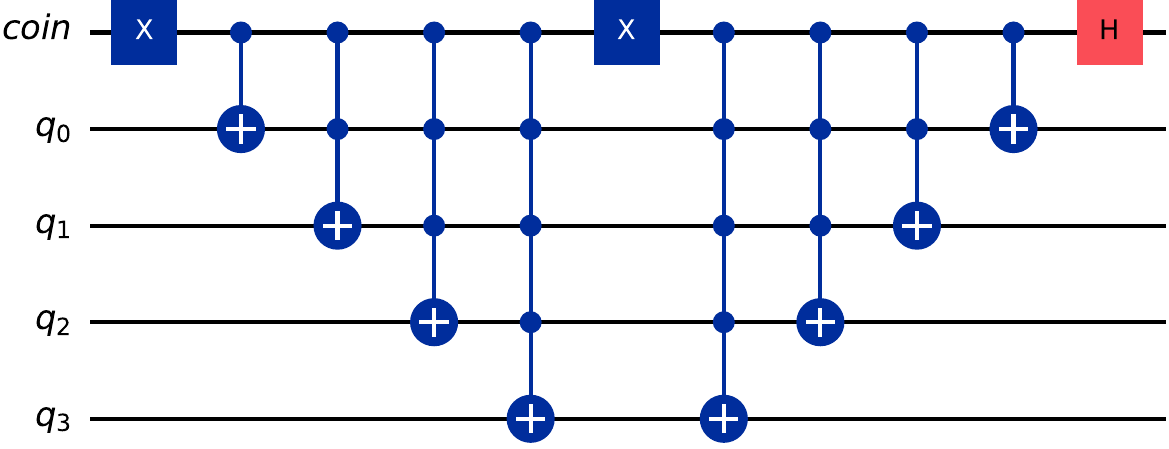}
\caption{Circuit for $U^\dagger$}
\label{fig:fig2}
\end{subfigure}
\caption{Quantum circuit for the evolution operator $U$ and $U^\dagger$.}
\label{fig:side_by_side}
\end{figure}

%It illustrates the evolution of the two initial states under the operators $U^5$ and $U^4$, along with their final measured position distributions.

The simulation results in Fig. 9 demonstrate that the final state's probability distribution is highly sensitive to both the initial position (derived from the previous block’s hash) and the number of walk steps (encoding the current block’s transactions). As illustrated in Figs. 9(d) and 9(f), even minor changes in the initial conditions or evolution dynamics cause drastic alterations in the final distribution, including the locations and amplitudes of probability peaks. This high sensitivity forms the basis of our scheme's tamper-resistance, as any unauthorized modification would result in a verifiably incorrect final state.

To complete the verification, block integrity is checked via inverse evolution. For a valid block, such as the one in Fig. 9(d) generated from initial state $|6\rangle$ with 5 walk steps ($U^5$), applying the inverse operator $(U^\dagger)^5$ deterministically restores the system to $|6\rangle|c_0\rangle$. A position measurement then yields $|6\rangle$, confirming the block. However, if the block's data has been altered (i.e., a different initial state or number of steps was used), applying the correct inverse operation $(U^\dagger)^5$ will fail to return the system to the expected initial state. Consequently, the measurement outcome will not be $|6\rangle$, revealing the tampering.

\begin{table*}[!htp]
\centering 
\caption{Comparisons between other quantum blockchain protocols }
\label{tab:1}       
\begin{tabular}{ccccccc}
\hline
\noalign{\smallskip}
    &   Structure of     &   Block linking           &    Immutability          &  Consensus    &   Consensus time &  Byzantine \\ 
    &    the chain   &    Method                           &   of  nodes        &  mechanism    &   complexity   & fault tolerance \\
\noalign{\smallskip} \hline \noalign{\smallskip} 
 Ref. \cite{15}   &   Classical chain   &   Hash function       &  No    &   Byzantine agreement &   $O(n^{f+1})$ & Yes ($\frac{n}{3}$) \\
 \noalign{\smallskip} 
 Ref. \cite{16}   &   GHZ state    &   Entanglement        &  No    &   $\theta$-protocol &   $O(n^2)$ &   No ($0$)\\
 \noalign{\smallskip} 
 Ref. \cite{19}   &   Weighted hypergraph    &   Entanglement       &  N/A    &  Relative phase &    $O(n)$ &  No ($0$) \\
                  &    states    &                     &                           &  consensus &          &     \\
 \noalign{\smallskip} 
 Ref. \cite{26}   &  Weighted hypergraph     &   Entanglement         &  Yes    &    Quantum  voting     &    $O(n)$ &   Yes ($\frac{n}{2}$)\\
                  &   or graph states  &                              &         &     protocol         &     &\\
 \noalign{\smallskip} 
 Our scheme  &    Quantum walks   &   Walk evolution         &   Yes          &   Weighted quantum &    $O(n)$ & Yes ($\frac{n}{2}$)\\
             &                    &                            &                &   voting protocol  &           &\\
\hline
\noalign{\smallskip}
\end{tabular}
\begin{tablenotes}
\item\ \ \  Note: $f$ denotes the number of faulty nodes.
\end{tablenotes}
\end{table*}
In summary, our simulations confirm the practical feasibility of the proposed quantum blockchain framework by validating its two core components. The results verify the correctness of the weighted quantum voting consensus and demonstrate an inherent tamper-detection mechanism in block construction based on quantum walks.

\section{Comparisons and discussion} 
This section presents a comparative analysis between our proposed scheme and representative quantum blockchain approaches. Although only a few of quantum blockchain frameworks have been introduced so far, the works in Refs. \cite{15,16,19,26} serve as benchmarks for evaluation. The key distinctions between these schemes and our proposal are summarized in Table II, followed by a detailed discussion of their structural and functional differences.

Firstly,  regarding block construction, existing quantum blockchain approaches~\cite{16,19,26} typically rely on multi-particle entangled states, such as weighted graph state or hypergraph state, to encode block information. The integrity of the entire chain thus depends on preserving global entanglement across all blocks. While such designs provide strong theoretical correlations, they face severe practical challenges: as the number of blocks $N$ increases, the resources required to prepare and protect the corresponding $N$-particle entangled states against decoherence scale exponentially, resulting in a significant scalability bottleneck. In contrast, our approach employs quantum walks to sequentially encode block data. Each new block is generated through local operations on the current state vector, avoiding the need for persistent long-range entanglement. The inherent sensitivity of quantum walks to initial conditions and evolution steps ensures block integrity,  and it enables tampering detection through inverse evolution. This design effectively removes the reliance on fragile entanglement, significantly enhancing both scalability and feasibility

Secondly, the proposed scheme introduces a weighted quantum voting mechanism for achieving consensus among nodes. In our design, the voting power of each participant is proportional to its assigned weight (e.g., stake or reputation), which reflects the characteristics of nodes in actual blockchain systems.
This weighted structure overcomes the uniformity constraints of symmetric voting, making decisions more realistic and fair. Our weighted voting protocol allows each voter to independently verify the correct calculation of their vote based on publicly available information, ensuring the fairness of the voting process. Furthermore, similar to the voting-based consensus in Ref. \cite{26}, our scheme elects representative nodes through quantum voting with a time complexity of $O(n)$, ensuring efficient blockchain operation over a period of time. Regarding Byzantine fault tolerance (BFT), our scheme and that of \cite{15, 26} provide BFT, whereas \cite{16} and \cite{19} do not. The time complexity and BFT rates of other schemes are summarized in Table II.

Thirdly, to evaluate the feasibility of the proposed quantum blockchain, we conducted simulations of two core components: quantum block construction and weighted quantum voting. The results confirm correctness, feasibility, and inherent tamper resistance. Specifically, the immutability of a block is verified by applying the inverse of its corresponding quantum walk operator. Due to the unitarity of quantum mechanics, if a block's state is tampered with, the reverse evolution will fail to restore it to the original state. Compared to classical blockchain schemes, which rely on the computational difficulty of finding hash collisions, our approach provides a more fundamental guarantee of block security. Furthermore, in contrast to previous entanglement-based schemes that require verifying the integrity of a large graph state after layer-by-layer unitary operations, our method is more direct and efficient.

In summary, our framework replaces fragile multipartite entanglement with quantum walks as the core mechanism for block construction and chain integrity. This design avoids the scalability bottlenecks of entanglement-based schemes. In addition, the proposed weighted quantum voting consensus more faithfully reflects the diversity among nodes, enabling fairer and more efficient decision-making. The proposed framework provides a secure, scalable, and feasible foundation for quantum blockchain systems.

\section{Conclusion}
In this work, we propose a quantum blockchain framework based on quantum walks, replacing the conventional entanglement-dependent structure, thereby  overcoming the scalability bottlenecks tied to maintaining long-range entanglement.  To link blocks and embed data, the framework maps the previous block's hash to the walk's initial condition and the current block's transactions to the evolution steps. Furthermore, we present a weighted quantum voting consensus mechanism that better captures the diversity of node weights, enabling more fair and efficient decision-making. Finally, circuit-level simulations are conducted to validate the correctness and feasibility of the proposed scheme. Overall, this framework provides a scalable alternative to entanglement-based quantum blockchains and offers new perspectives for the design of future quantum blockchains.
%

%\appendices
%\section{Proof of the First Zonklar Equation}
%Appendix one text goes here.
%
%% you can choose not to have a title for an appendix
%% if you want by leaving the argument blank
%\section{}
%Appendix two text goes here.

% use section* for acknowledgment
\section*{Acknowledgment}
This work was supported in part by the National Natural Science Foundation of China under Grant Nos. 62501523 and 61871347,
as well as the Zhejiang Province Selected Funding for Postdoctoral Research Projects.

% Can use something like this to put references on a page
% by themselves when using endfloat and the captionsoff option.
\ifCLASSOPTIONcaptionsoff
  \newpage
\fi

% if you will not have a photo at all:
%\begin{IEEEbiographynophoto}{John Doe}
%Biography text here.
%\end{IEEEbiographynophoto}

% insert where needed to balance the two columns on the last page with
% biographies
%\newpage

%\begin{IEEEbiographynophoto}{Jane Doe}
%Biography text here.
%\end{IEEEbiographynophoto}

% You can push biographies down or up by placing
% a \vfill before or after them. The appropriate
% use of \vfill depends on what kind of text is
% on the last page and whether or not the columns
% are being equalized.

%\vfill

\end{document}